\documentclass[a4paper,11pt]{article}
\usepackage{cite}
\usepackage{hyperref}
\usepackage[top=1 in,bottom=1 in,left=1.0 in,right=1.0 in]{geometry}
\hypersetup{hidelinks}
\usepackage{amsthm}
\usepackage{amssymb,amsmath}
\usepackage{appendix}
\usepackage{xcolor}

\numberwithin{equation}{section}

\newtheorem{lemma}{Lemma}
\newtheorem{theorem}{Theorem}
\newtheorem{remark}{Remark}

\begin{document}

\title{{A unified approach to the AKNS, DNLS, KP and mKP hierarchies
in the anti-self-dual Yang-Mills reduction}}

\author{Shangshuai Li$^{1,2}$\footnote{Email: lishangshuai@nbu.edu.cn}
	~~ Ken-ichi Maruno$^{3}$\footnote{Email: kmaruno@waseda.jp}
	~~ Da-jun Zhang$^{4,5}$\footnote{Corresponding author. Email: djzhang@staff.shu.edu.cn}
	~~\\
	{\small $^1$School of Mathematics and Statistics, Ningbo University, Ningbo 315211, China} \\
	{\small $^2$Graduate School of Mathematics, Nagoya University, Nagoya 464-8602, Japan} \\
	{\small $^{3}$Department of Applied Mathematics, Faculty of Science and Engineering, } \\
	{\small Waseda University, Tokyo 169-8555, Japan} \\
	{\small $^{4}$Department of Mathematics, Shanghai University, Shanghai 200444, China} \\
	{\small $^{5}$Newtouch Center for Mathematics of Shanghai University,  Shanghai 200444, China} 	
}

\maketitle

\begin{abstract}
We show a unified approach to the Ablowitz-Kaup-Newell-Segur (AKNS) hierarchy
and
the unreduced derivative nonlinear Schr\"odinger (DNLS) hierarchies
(including the Kaup-Newell, Chen-Lee-Liu,  Gerdjikov-Ivanov and a generalized DNLS),
together with their multi-component extensions,
in the framework of the anti-self-dual Yang-Mills (ASDYM) reduction.
By restricting the gauge group to GL(2), the Kadomtsev-Petviashvili (KP) and modified KP (mKP) hierarchies are formulated in the ASDYM reduction via squared eigenfunction symmetry constraints.
In this case, the bilinearization of the generalized DNLS equations
can also be understood through this reduction.
Finally,  Gram-type exact solutions for the relevant equations
are presented in terms of quasi-determinants.

\begin{description}
\item[Keywords:] anti-self-dual Yang-Mills equation; reduction; gauge transformation; hierarchy structure; exact solution
\end{description}

\end{abstract}



\section{Introduction}\label{sec-1}

The anti-self-dual Yang-Mills (ASDYM) equation is a fundamental model in quantum field theory and various branches of mathematical physics \cite{Mason-Book}.
Since the 1970s, the ASDYM equation has been intensively studied, revealing numerous properties,
such as the connection to complex vector bundles \cite{Donaldson-1985, Ward-1977, Gu-1981}, the existence of instanton solutions \cite{Atiyah-1977, Wilczek-1977, Atiyah-1978, Corrigan-1978}, the Painlev\'e property \cite{Jimbo-1982, Ward-1984} and so on.

As an integrable system, the ASDYM equation has been investigated using a variety of integrable methods, such as the inverse scattering transformation \cite{Belavin-1978}, B\"acklund transformation based on the Riemann-Hilbert problem \cite{Ueno-1982}, a direct method inspired by the Sato theory \cite{Takasaki-1984}, bilinear method \cite{Sasa-1988,Ohta-2001,Ohta-2024}, Darboux transformation \cite{Nimmo-2000,Hamanaka-2020,Hamanaka-2022,Huang-2021}, direct linearization approach \cite{LSS-2025}, bi-differential graded algebra approach  \cite{Muller-Hoissen-2024} and so on.

The fact that the ASDYM equation can be reduced to many lower-dimensional classical integrable systems
(see Ward's conjecture
in \cite{Ward-1985}) has motivated a trend in searching for such integrable reductions.
For example, the Korteweg-de Vries (KdV) equation and the nonlinear Schr\"odinger (NLS) equation can be obtained from the SL(2) ASDYM equation \cite{Mason-1989}. The relationship between the SL($N$) ASDYM equation and the Gel'fand-Dickey hierarchy (which includes the KdV equation) has been established in \cite{Bakas-1992}. Furthermore, reduction to the Kadomtsev-Petviashvili (KP) hierarchy are possible by allowing the gauge potentials to be operator-valued \cite{Schiff-1992}. Ablowitz et al. discussed the derivation of integrable hierarchies from the ASDYM hierarchy, providing examples such as the $N$-wave system, and the KdV, NLS, Davey-Stewartson (DS), and KP equations \cite{Ablowitz-1993}. Ge et al. derived reductions to the sine-Gordon and Liouville equations \cite{Ge-1994}.

In addition to the KP and DS equations, the ASDYM equation admits reductions to other (2+1)-dimensional equations,
including the modified KP (mKP) equation, a (2+1)-dimensional Gardner equation, a generalized DS equation \cite{Chakravarty-1995}, and the Calogero-Bogoyavlenskii-Schiff (CBS) equation \cite{Schiff-1992-2}.
Furthermore, the Ward's chiral model in (2+1) dimensions \cite{Ward-1988,Ward-1988-2} and a (2+1)-d relativistic-invariant field theory model \cite{Manakov-1981} are obtainable from the $J$-matrix formulation;
a hierarchy of (2+1)-d NLS equation and its relationship with the ASDYM equation have been discussed in \cite{Kakei-2002}.

For one-dimensional equations, the Painlev\'e equations arise through the reductions of conformal symmetries
\cite{Mason-Book,Mason-1993}. Besides, B\"acklund transformations of the ASDYM equation can be reduced to obtain some discrete integrable systems \cite{Benincasa-thesis}.
Moreover, Ward's conjecture also extends to non-commutative space \cite{Hamanaka-2006}. Comprehensive reviews of the ASDYM reductions can be found in \cite{Mason-Book, Ablowitz-2003}, along with their extensive references.

A recent development is the formulation of the Cauchy matrix schemes of the ASDYM equations \cite{LSS-2022, LSS-2023, LSS-2025-2}, which facilitate connections between the ASDYM equations and other lower-dimensional integrable systems. As a result, a reduction from the ASDYM matrix formulations to the Fokas-Lenells (FL) equation was found in \cite{LSS-2024}. The FL equation is the first negative member of the potential Kaup-Newell (KN) hierarchy \cite{Fokas-1995,Lenells-2009,Lenells-2009-2},
which is commonly known as the (unreduced) derivative nonlinear Schr\"odinger (DNLS) hierarchy.

The reduction from the ASDYM to the FL equation in \cite{LSS-2024}
suggests the existence of further hidden structures that can be exploited for reductions.
In this paper, we aim to reveal these structures
and establish the connection between the ASDYM hierarchy and other integrable hierarchies,
especially the Ablowitz-Kaup-Newell-Segur (AKNS) hierarchy,
the unreduced DNLS hierarchies
(including the KN, Chen-Lee-Liu (CLL),  Gerdjikov-Ivanov (GI) and a generalized DNLS),
and the KP and mKP hierarchies.
We will make use of the $J$-matrix formulation and $K$-matrix formulation of the ASDYM
hierarchy (see equation \eqref{J-matrix-formulation} and \eqref{K-matrix-formulation}).
Here we specially mention that
by restricting the gauge group to GL(2),
the KP hierarchy and mKP hierarchy are respectively defined by
the $(1,1)$-th entries of $K$ and $J$, which will be explained in Sec.\ref{sec-4}.
It has long been believed that the KP equations cannot be derived from the ASDYM equations via reduction unless operator-valued gauge fields are permitted \cite{Schiff-1992-2,Ablowitz-2003}.
In contrast, we provide an operator-free reduction to the KP and mKP hierarchies from the ASDYM equations.

This paper is organized as follows: In Section \ref{sec-2}, we briefly review the ASDYM hierarchy and several classical integrable hierarchies.
The transformations between these hierarchies, realized within the ASDYM framework, are established in Section \ref{sec-3}.
 In Section \ref{sec-4}, we summarize the results for the GL(2) case, where we discuss the Lax pairs of the AKNS hierarchy, the reduction to the KP and mKP hierarchies, and the bilinear transformation of the generalized DNLS system.
For the completeness of this paper,
we provide the procedure of constructing exact solutions in Section \ref{sec-5},
and show how ASDYM solutions map  to those of the reduced systems.
The final part is devoted to the concluding remarks.

\section{Preliminary: hierarchy structures of some integrable systems}\label{sec-2}

In this section, we will review the hierarchy structures of the ASDYM equations
and several representative integrable equations.

\subsection{The ASDYM equations}\label{sec-2-1}

In this paper, we mainly consider the ASDYM hierarchy of the following form.
Let $\mathfrak{g}$ be the Lie algebra of a Lie group $G$.
We assume the following zero-curvature system holds \cite{Takasaki-1990,Nakamura-1988}:
\begin{align}\label{zero-curvature-system}
	[\partial_{n+1}-\lambda\partial_{n}+A_{n+1},\partial_{m+1}-\lambda\partial_{m}+A_{m+1}]=0,
\end{align}
i.e.,
\begin{subequations}\label{zero-curvature-system-new}
	\begin{align}
		\label{Lax-1}&\partial_mA_{n+1}-\partial_nA_{m+1}=0, \\
		\label{Lax-2}&\partial_{m+1}A_{n+1}-\partial_{n+1}A_{m+1}-[A_{n+1},A_{m+1}]=0,
	\end{align}
\end{subequations}
which is the compatible condition of the associated linear system:
\begin{align}
	(\partial_{n+1}-\lambda\partial_{n}+A_{n+1})\Psi=0,~~~~n\in\mathbb Z,
\end{align}
or alternatively
\begin{subequations}\label{ASDYM-linear-system}
	\begin{align}
		(\partial_n-\lambda^n\partial_0+\sum_{i=0}^{n-1}\lambda^iA_{n-i})\Psi=0,~~~~n\in\mathbb Z^+,\\
		(\partial_n-\lambda^n\partial_0-\sum_{i=n}^{-1}\lambda^iA_{n-i})\Psi=0,~~~~n\in\mathbb Z^-.
	\end{align}
\end{subequations}
Here $\lambda$ is the spectral parameter, $\Psi=\Psi(\lambda)$ is the eigenfunction and $A_{n+1} \in \mathfrak{g}$ with $n\in\mathbb Z$ are gauge potentials.
This system depends on a set of infinite independent complex
variables $\mathbf t:=\{t_{n}\}$. We denote the corresponding partial derivatives by $\partial_{n}:=\partial/\partial t_{n}$.
The gauge potentials $A_{n+1}$ can be expressed in terms of either $J$-matrix or $K$-matrix via the definitions:
\begin{align}\label{gauge-potential}
	A_{n+1}:=-(\partial_{n+1}J)J^{-1}=\partial_{n}K.
\end{align}
Substituting this expression into the zero-curvature system \eqref{zero-curvature-system-new} yields the following  two matrix formulations of the ASDYM hierarchy:
\begin{itemize}
	\item $J$-matrix formulation (the Yang equation \cite{Yang-1977,Pohlmeyer-1980}):
	\begin{align}\label{J-matrix-formulation}
		\partial_{m}((\partial_{n+1}J)J^{-1})-\partial_{n}((\partial_{m+1}J)J^{-1})=0.
	\end{align}
	\item $K$-matrix formulation (the Chalmers-Siegel equation \cite{Parkes-1992,Chalmers-1996}):
	\begin{align}\label{K-matrix-formulation}
		\partial_{m+1}\partial_{n}K-\partial_{n+1}\partial_{m}K-[\partial_{n}K,\partial_{m}K]=0.
	\end{align}
\end{itemize}

\begin{remark}\label{Rem-0}
The compatible expression \eqref{gauge-potential} for the gauge potentials $\{A_j\}$
is important in our research
because  both the $J$-matrix formulation  \eqref{J-matrix-formulation} and
the $K$-matrix formulation  \eqref{K-matrix-formulation}
are the consequences of \eqref{gauge-potential}.
In some reductions, we need to start from the gauge potential expression \eqref{gauge-potential}
rather than the above two equations.
We will see the advantage in so doing.
\end{remark}

\subsection{The  AKNS equations}\label{sec-2-2}

The well-known NLS equation is one of the most important integrable equations,
arising in a wide range of  physical contexts \cite{Zakharov-1972}:
\begin{align}\label{NLS}
	\mathrm iu_\tau+u_{\zeta\zeta}+2\delta|u|^2u=0,~~~~\delta=\pm1.
\end{align}
Here $\mathrm{i}^2=-1$, $|u|^2=u\bar u$
where bar represents complex conjugate,
and the parameter $\delta$ distinguishes between the focusing ($\delta=1$) and defocusing ($\delta=-1$) cases of the nonlinearity.

Now, let us forget about the complex structure of \eqref{NLS} and consider the following system (the NLS equation is recovered by introducing $(\zeta,\tau):=(\mathrm ix,\mathrm it_2)$ as real coordinates and requiring $u:=r=\delta \bar q$):
\begin{subequations}\label{NLS-couple}
	\begin{align}
		&r_{t_2}=r_{xx}-2qr^2, \\
		&q_{t_2}=-q_{xx}+2q^2r,
	\end{align}
\end{subequations}
which is known as the unreduced NLS system,
belonging to the positive   AKNS hierarchy
(see \cite{AKNS-1974} or Chapter 3.3 of \cite{CDY-book}):
\begin{align}\label{AKNS-h}
	\begin{pmatrix}
		r \\
		q
	\end{pmatrix}_{t_{n}}
	=R^n H_0, ~~~~
	R:=
	\begin{pmatrix}
		\partial_x-2r\partial^{-1}_x q & -2r\partial^{-1}_x r \\
		2q\partial^{-1}_x q & -\partial_x +2q\partial^{-1}_x r
	\end{pmatrix},
	~~~~ H_0=\begin{pmatrix}
		r \\
		-q
	\end{pmatrix},
\end{align}
where $\partial^{-1}_x$ is the integration operator defined as $\partial_x^{-1}f=\frac{1}{2}(\int^x_{-\infty}-\int_{x}^{\infty})f(x')\mathrm dx'$, and $n\in\mathbb Z$.
The recursion operator $R$ can be also rewritten as
\begin{align}
	R:=
	\sigma_3\partial_x+2
	\begin{pmatrix}
		-r \\ q
	\end{pmatrix}\partial_x^{-1}\big(q~\, r\big),~~~~\sigma_3:=\mathrm{Diag}(1,-1).
\end{align}
This hierarchy can be alternatively expressed in a recursive form
\begin{align}\label{AKNS-hn}
	\begin{pmatrix}
		r \\
		q
	\end{pmatrix}_{t_{n+1}}
	=R \begin{pmatrix}
		r \\
		q
	\end{pmatrix}_{t_{n}}, ~~~~
	\begin{pmatrix}
		r \\
		q
	\end{pmatrix}_{t_{0}}=\begin{pmatrix}
		r \\
		-q
	\end{pmatrix},
\end{align}
which expands to the following system:
\begin{subequations}\label{AKNS}
	\begin{align}
		\label{AKNS-1}&r_{t_{n+1}}=r_{x,t_n}-2r\partial_x^{-1}(qr)_{t_n},~~~~r_{t_0}=r, \\
		\label{AKNS-2}&q_{t_{n+1}}=-q_{x,t_n}+2q\partial_x^{-1}(rq)_{t_n},~~~~q_{t_0}=-q.
	\end{align}
\end{subequations}

\subsection{The DNLS equations and gauge transformations}\label{sec-2-3}

The NLS-type equations with derivative nonlinearities are referred to as
the derivative nonlinear Schr\"odinger (DNLS) equations.
One example of the DNLS equations, proposed by Gerdjikov and Ivanov (GI), is presented as \cite{Gerdjikov-1983}:
\begin{align}\label{GI-equation}
	\mathrm iu_\tau+u_{\zeta\zeta}+2\mathrm i\delta u^2\bar u_\zeta+2|u|^4u=0.
\end{align}
The  unreduced GI system reads
\begin{subequations}
	\begin{align}
		&p_{t_{2}}=p_{xx}-2p^2q_x-2p^3q^2,\\
		&q_{t_{2}}=-q_{xx}-2q^2p_x+2p^2q^3,
	\end{align}
\end{subequations}
which yields the GI equation \eqref{GI-equation} by
taking $(\zeta,\tau):=(\mathrm ix,\mathrm it_2)$ as real coordinates
and the conjugate reduction $u:=p=\delta \bar q$.
The above system is the first nonlinear positive member of the GI hierarchy:
\begin{align}\label{GI-hierarchy}
	\begin{pmatrix}
		p \\
		q
	\end{pmatrix}_{t_{n+1}}=\mathcal R
	\begin{pmatrix}
		p \\
		q
	\end{pmatrix}_{t_{n}},~~~~
	\begin{pmatrix}
		p \\
		q
	\end{pmatrix}_{t_{0}}=
	\begin{pmatrix}
		p \\
		-q
	\end{pmatrix},
\end{align}
where $\mathcal R$ is the recursion operator \cite{Fan-2000}:
\begin{align}\label{GI-recursion}
	\mathcal R:=
	\begin{pmatrix}
		\partial_x - 2p \partial^{-1}_x (q_x + 2pq^2) & -2p^2 + 2p \partial^{-1}_x (p_x - 2p^2q) \\
		-2q^2 + 2q \partial^{-1}_x (q_x + 2pq^2) & -\partial_x - 2q \partial^{-1}_x (p_x - 2p^2q)
	\end{pmatrix}.
\end{align}
The expansion of \eqref{GI-hierarchy} yields
\begin{subequations}\label{GI}
	\begin{align}
		\label{GI-1}&p_{t_{n+1}}=p_{x,t_n}-2p\partial^{-1}_x(pq_x)_{t_n}-2p\partial^{-1}_x(p^2q^2)_{t_n},
~~~~p_{t_0}=p, \\
		\label{GI-2}&q_{t_{n+1}}=-q_{x,t_n}-2q\partial^{-1}_x(p_xq)_{t_n}+2q\partial^{-1}_x(p^2q^2)_{t_n},
~~~~q_{t_0}=-q.
	\end{align}
\end{subequations}

\begin{remark}\label{Rem-1}
We will show that the AKNS hierarchy \eqref{AKNS} and the GI hierarchy \eqref{GI}
share the same $q$-function in our construction.
In fact, there is a Riccati-type Miura transformation between the two systems (cf. \cite{Kakei-2004}):
	\begin{align}\label{Riccati-Miura-transformation}
		r=-p_x+p^2q.
	\end{align}
The equation \eqref{GI-2} is a direct result of \eqref{AKNS-2} in light of \eqref{Riccati-Miura-transformation}.
Since the $t_{n+1}$-derivative of $r$ can be represented as
	\begin{align}
		r_{t_{n+1}}=(2pq-\partial_x)p_{t_{n+1}}+p^2 q_{t_{n+1}},
	\end{align}
we are led to \eqref{GI-1} by substituting \eqref{Riccati-Miura-transformation} into \eqref{AKNS-1}
and utilizing \eqref{GI-2}.	
This relation also facilitated the investigations on the FL equation from the AKNS negative flow \cite{Ye-2023,Muller-Hoissen-2025}, which was recently realized in the ASDYM reduction \cite{LSS-2024}.
\end{remark}

The DNLS equations are of three types, other than the DNLS equation of GI-type \eqref{GI-equation},
there are the Chen-Lee-Liu (CLL) DNLS equation \cite{Chen-1979}:
\begin{align}\label{CLL-equation}
	\mathrm iu_{\tau}+u_{\zeta\zeta}-2\mathrm i\delta|u|^2u_\zeta=0,
\end{align}
and the Kaup-Newell (KN) DNLS equation \cite{Kaup-1978}:
\begin{align}\label{KN-equation}
	\mathrm iu_{\tau}+u_{\zeta\zeta}-2\mathrm i\delta(|u|^2u)_\zeta=0.
\end{align}
The three models are all integrable and related via gauge transformations \cite{Wadati-1982}.
In fact, the method of gauge transformation can be further applied to derive a generalized derivative nonlinear Schr\"odinger (GDNLS) equation. Through the transformation (here we use $u_{\text{[GDNLS]}}$ and $u_{\text{[GI]}}$ to distinguish the variables in different DNLS equations):
\begin{align}\label{gauge-GI-GDNLS}
	u_{\text{[GDNLS]}}=u_{\text{[GI]}}\exp\left(\mathrm i\gamma\delta\int\left|u\right|^2_{\text{[GI]}}\mathrm d\zeta\right),~~~~\gamma\in\mathbb C,
\end{align}
one obtains the GDNLS equation (here we denote $u=u_{\text{[GDNLS]}}$) \cite{Kundu-1984,Kundu-1987,Clarkson-1987}:
\begin{align}\label{GDNLS-equation}
	\mathrm iu_\tau+u_{\zeta\zeta}-2\mathrm i\gamma\delta |u|^2u_\zeta
-2\mathrm i(\gamma-1)\delta u^2\bar u_\zeta+(\gamma-1)(\gamma-2)|u|^4u=0.
\end{align}
It becomes the GI equation \eqref{GI-equation} when $\gamma=0$,
the CLL equation \eqref{CLL-equation} when $\gamma=1$
and the KN equation \eqref{KN-equation} when $\gamma=2$.
Nevertheless, these three equations are typically investigated independently, as the gauge transformations linking them involve complex integrations that become difficult to compute in multi-soliton case.

The hierarchy structure of \eqref{GDNLS-equation} can be investigated by applying the gauge transformation to \eqref{GI}. Inspired from \eqref{gauge-GI-GDNLS}, the gauge factor is introduced as
\begin{align}\label{gauge-factor}
	s=\exp(-\partial_x^{-1}(pq)),
\end{align}
which satisfies the following dynamics
\begin{subequations}\label{s-dynamics}
	\begin{align}
		&\frac{s_x}{s}=-pq, \label{s-dynamics-x}  \\
		& \frac{s_{t_n}}{s}=-\partial^{-1}_x(pq)_{t_n}, \\
		&\frac{s_{t_{n+1}}}{s}=-\partial_x^{-1}(p_xq)_{t_n}+\partial_x^{-1}(p^2q^2)_{t_n}+pq_{t_n}.
	\end{align}
\end{subequations}
The new variables are introduced as
\begin{equation}\label{uv}
(u,v):=(s^{\gamma}p,s^{-\gamma}q),
\end{equation}
which transforms the GI hierarchy \eqref{GI} to the GDNLS hierarchy
\begin{subequations}\label{GDNLS}
	\begin{align}
		u_{t_{n+1}}&=\,u_{x,t_n}+\gamma u\partial^{-1}_x(u_xv)_{t_n}
+\gamma u_x\partial^{-1}_x(uv)_{t_n} \notag \\
		\label{GDNLS-1}&~~~~~~~~~~+2(\gamma-1)u\partial^{-1}_x(uv_x)_{t_n}
-(\gamma-1)(\gamma-2)u\partial^{-1}_x(u^2v^2)_{t_n}, \\
		v_{t_{n+1}}&=-v_{x,t_n}+\gamma v\partial^{-1}_x(uv_x)_{t_n}
+\gamma v_z\partial^{-1}_x(uv)_{t_n} \notag \\
		\label{GDNLS-2}&\,~~~~~~~~~~~+2(\gamma-1)v\partial^{-1}_x(u_xv)_{t_n}
+(\gamma-1)(\gamma-2)v\partial^{-1}_x(u^2v^2)_{t_n},
	\end{align}
\end{subequations}
where $u_{t_0}=u$ and $v_{t_0}=-v$. The GDNLS hierarchy
with recursive coordinates can be generated by the action of a recursion operator $\mathcal M$
on the preceding flow
\begin{align}\label{GDNLS-recursion}
	\begin{pmatrix}
		u \\
		v
	\end{pmatrix}_{t_{n+1}}=\mathcal M
	\begin{pmatrix}
		u \\
		v
	\end{pmatrix}_{t_n},~~~~
	\begin{pmatrix}
		u \\
		v
	\end{pmatrix}_{t_0}=
	\begin{pmatrix}
		u \\
		-v
	\end{pmatrix},
\end{align}
where the recursion operator $\mathcal M=\{\mathcal M_{ij}\}$ is given by
\begin{subequations}
	\begin{align}
		\mathcal M_{11} &= \partial_x + \gamma uv + (\gamma-2)u\partial^{-1}_x v_x
+ \gamma u_x \partial^{-1}_x v - 2(\gamma-1)(\gamma-2) u\partial^{-1}_x uv^2, \\
		\mathcal M_{12} &= 2(\gamma-1)u^2 - (\gamma-2)u\partial^{-1}_x u_x
+ \gamma u_x \partial^{-1}_x u - 2(\gamma-1)(\gamma-2) u\partial^{-1}_x u^2v, \\
		\mathcal M_{21} &= 2(\gamma-1)v^2 - (\gamma-2)v\partial^{-1}_x v_x
+ \gamma v_x \partial^{-1}_x v + 2(\gamma-1)(\gamma-2) v\partial^{-1}_x uv^2, \\
		\mathcal M_{22} &= -\partial_x + \gamma uv + (\gamma-2)v\partial^{-1}_x u_x
+ \gamma v_x \partial^{-1}_x u + 2(\gamma-1)(\gamma-2) v\partial^{-1}_x u^2v.
	\end{align}
\end{subequations}

As a direct consequence of \eqref{GDNLS} and \eqref{GDNLS-recursion},
the recursive formulations and recursion operators for the CLL and KN hierarchies are provided without proof.
Taking $\gamma=1$ in \eqref{GDNLS} and define
\begin{equation}\label{pq-tilde}
(\tilde p,\tilde q):=(u|_{\gamma=1},v|_{\gamma=1})=(sp,s^{-1}q),
\end{equation}
it gives rise to the CLL hierarchy
\begin{subequations}\label{CLL}
	\begin{align}
		\label{CLL-1}&\tilde p_{t_{n+1}}=\tilde p_{x,t_n}+\tilde p\partial_x^{-1}(\tilde p_x\tilde q)_{t_n}
+\tilde p_x\partial_x^{-1}(\tilde p\tilde q)_{t_n}, ~~~~ \tilde p_{t_{0}}=\tilde p,\\
		\label{CLL-2}&\tilde q_{t_{n+1}}=-\tilde q_{x,t_n}+\tilde q\partial_x^{-1}(\tilde p\tilde q_x)_{t_n}
+\tilde q_x\partial_x^{-1}(\tilde p\tilde q)_{t_n}, ~~~~ \tilde q_{t_{0}}=-\tilde q.
	\end{align}
\end{subequations}
The recursion operator for CLL hierarchy is presented as
\begin{align}\label{CLL-recursion}
	\widetilde{\mathcal R}:=
	\begin{pmatrix}
		\partial_x+\tilde p\tilde q+\tilde p_x\partial_x^{-1}\tilde q-\tilde p\partial_x^{-1}\tilde q_x &
\tilde p_x\partial_x^{-1}\tilde p+\tilde p\partial_x^{-1}\tilde p_x \\
		\tilde q_x\partial_x^{-1}\tilde q+\tilde q\partial_x^{-1}\tilde q_x & -\partial_x
+\tilde p\tilde q+q_x\partial^{-1}_x\tilde p-\tilde q\partial_x^{-1}\tilde p_x
	\end{pmatrix}.
\end{align}
Taking
\begin{equation}\label{pq-hat}
(\hat p,\hat q):=(u|_{\gamma=2},v|_{\gamma=2})=(s^2p,s^{-2}q),
\end{equation}
it yields the KN hierarchy
\begin{subequations}\label{KN}
\begin{align}
&\hat p_{t_{n+1}}=\hat p_{x,t_n}+2(\hat p\partial_x^{-1}(\hat p\hat q)_{t_n})_x, ~~~~ \hat p_{t_{0}}=\hat p,\\
&\hat q_{t_{n+1}}=-\hat q_{x,t_n}+2(\hat q\partial_x^{-1}(\hat p\hat q)_{t_n})_x, ~~~~\hat q_{t_{0}}=-\hat q,
\end{align}
\end{subequations}
with the associated recursion operator
\begin{align}
	\widehat{\mathcal R}:=
	\begin{pmatrix}
		\partial_x+2\hat p\hat q+2\hat p_x\partial^{-1}_x\hat q & 2\hat p^2+2\hat p_x\partial^{-1}_x\hat p \\
		2\hat q^2+2\hat q_x\partial^{-1}_x\hat q & -\partial_x+2\hat p\hat q+2\hat q_x\partial^{-1}_x\hat p
	\end{pmatrix}.
\end{align}
The recursion operator of KN hierarchy can also be written as
\begin{align}
	\widehat{\mathcal R}:=
	\begin{pmatrix}
		\partial_x+2\partial_x(\hat p\partial_x^{-1}\hat{q}) & 2\partial_x(\hat p\partial_x^{-1}\hat{p}) \\
		2\partial_x(\hat q\partial_x^{-1}\hat{q}) & -\partial_x+2\partial_x(\hat q\partial_x^{-1}\hat{p})
	\end{pmatrix},
\end{align}
or alternatively
\begin{align}
	\widehat{\mathcal R}:=
	\sigma_3\partial_x+2\partial_x
	\begin{pmatrix}
		\hat p \\ \hat q
	\end{pmatrix}\partial_x^{-1}\big(\hat q~\,\hat p\big),~~~~\sigma_3:=\mathrm{Diag}(1,-1).
\end{align}

\subsection{The KP hierarchy and squared eigenfunction symmetry constraint }\label{sec-2-4}

The integrability of the KP hierarchy can be described through utilizing the pseudo-differential operator \cite{Sato-1981,Date-1982,Ohta-1988}. It is formulated as the compatible condition of the linear system
\begin{subequations}\label{KP-linear}
	\begin{align}
		&L\phi=\lambda\phi,~~~~L:=\partial_x+u_2\partial^{-1}_x+u_3\partial^{-2}_x+\cdots, \\
		&\phi_{t_m}=B_m\phi,~~~~B_m:=(L^m)_{\geq0},~~~~m\geq1.
	\end{align}
\end{subequations}
The compatible condition of \eqref{KP-linear} indicates the Lax representation of KP hierarchy
\begin{subequations}
	\begin{align}
		&\partial_{t_m}L=[B_m,L], \\
		\label{ZS-KP}&\partial_{t_m}B_n-\partial_{t_n}B_m+[B_n,B_m]=0,
	\end{align}
\end{subequations}
which gives rise to the KP hierarchy (here we denote $\hat u=u_2$)
\begin{subequations}\label{KP-hierarchy}
	\begin{align}
		&\hat u_{t_1}=K_1=\hat u_x,\\
		&\hat u_{t_2}=K_2=\hat u_y, \\
		&\hat u_{t_3}=K_3=\frac{1}{4}\hat u_{xxx}+3\hat u\hat u_x+\frac{3}{4}\partial_x^{-1}\hat u_{yy}, \label{KP-eq}\\
		& ~~~~\cdots\cdots \nonumber
	\end{align}
\end{subequations}

The linear system \eqref{KP-linear} and its adjoint form are closely related to the AKNS hierarchy \eqref{AKNS-h}
under the squared eigenfunction symmetry constraint. We explain it by introducing the inner product
\begin{align}
	\langle f(x),g(x)\rangle:=\int_{-\infty}^{\infty}f(x)g(x)\mathrm dx,
\end{align}
where $f(x)$ and $g(x)$ tend to zero fast enough when $|x|\to \infty$.
It follows that  the adjoint operator of $O$, denoted as $O^*$, is defined by
\begin{align}
	\langle f(x),O g(x)\rangle=\langle O^*f(x),g(x)\rangle.
\end{align}
Thus, apart from \eqref{KP-linear}, one can alternatively consider its adjoint system
\begin{align}\label{dual-system}
	L^*\phi^*=\lambda\phi^*,~~~~	\phi^*_{t_m}=-B_m^*\phi^*,
\end{align}
where we use $\phi^*$ to denote the eigenfunction of \eqref{dual-system}. It can be proved that $(\phi\phi^*)_x$
is a symmetry of the KP hierarchy, conventionally called the squared eigenfunction symmetry \cite{Oevel-1993,Matsukidaira-1990}. Since $\hat u_x$ is also a symmetry of the KP hierarchy,
one may consider a symmetry constraint $\hat u_x+(\phi\phi^*)_x=0$, which indicates that
\begin{align}\label{u-qr}
	\hat u=-\phi\phi^*=-rq,~~~~(r,q):=(\phi,\phi^*).
\end{align}
It turns out that under the above constraint, the pseudo-differential operator $L$ can be represented as
\begin{align}
	L=\partial_x-r\partial_x^{-1}q.
\end{align}
This leads the linear system $\phi_{t_n}=B_n\phi$ and $\phi^*_{t_n}=-B_n^*\phi^*$ to a constrained form
\begin{subequations}
	\begin{align}
		&r_{t_n}=B_nr=\left((\partial_x-r\partial_x^{-1}q)^n\right)_{\geq0}r, \\
		&q_{t_n}=-B_n^*q=-\left((-\partial_x+q\partial_x^{-1}r)^n\right)_{\geq0}q,
	\end{align}
\end{subequations}
which was proved to be the positive AKNS hierarchy \eqref{AKNS-h} with $n\geq1$
\cite{KSS-PLA-1991,CL-PLA-1991,Xu-1992,CL-JPA-1992}.
Based on the above connection, once we have a solution pair $(r,q)$ of the AKNS hierarchy with $n\geq1$,
then $\hat u=-qr$ provides a solution for the KP hierarchy \eqref{KP-hierarchy}.

\begin{remark}\label{Rem-2}
	Based on the Miura transformation \eqref{Riccati-Miura-transformation} and the gauge transformation \eqref{gauge-factor}, the KP solution can be provided by (see also in \cite{CL-JPA-1992,Dimakis-2006})
	\begin{align}
		\hat u=\tilde p_x\tilde q=p_xq-p^2q^2=-qr,
	\end{align}
	where $(p,q)$ satisfy the GI hierarchy \eqref{GI}
and $(\tilde p,\tilde q)$ satisfy the CLL hierarchy \eqref{CLL}.
\end{remark}

\subsection{The mKP hierarchy and squared eigenfunction symmetry constraint }\label{sec-2-5}

Likewise, there is a similar result for the mKP hierarchy.
The mKP hierarchy is governed by the linear system
\begin{subequations}\label{mKP-linear}
	\begin{align}
		&\mathcal L\psi=\lambda\psi,~~~~\mathcal L:=\partial_x+v_0+v_1\partial_x^{-1}+\cdots, \\
		&\psi_{t_m}=\mathcal B_m\psi,~~~~\mathcal B_m:=(\mathcal L^m)_{\geq1},~~~~m\geq1,
	\end{align}
\end{subequations}
where the compatibility  gives rise to the mKP hierarchy (here we denote $\hat v=v_0$)
\begin{subequations}\label{mKP-hierarchy}
	\begin{align}
		&\hat v_{t_1}=\mathcal K_1=\hat v_x,\\
		&\hat v_{t_2}=\mathcal K_2=\hat v_y, \\
		\label{mKP}&\hat v_{t_3}=\mathcal K_3=\frac{1}{4}\hat v_{xxx}-\frac{3}{2}\hat v^2\hat v_x
+\frac{3}{2}\hat v_x\partial^{-1}_x\hat v_y+\frac{3}{4}\partial^{-1}_x\hat v_{yy}, \\
		& ~~~~\cdots\cdots \nonumber
	\end{align}
\end{subequations}
Equation \eqref{mKP} is known as the mKP equation.
The adjoint system of \eqref{mKP-linear} is given by
\begin{align}
	\mathcal L^*\psi^*=\lambda\psi^*,~~~~\psi^*_{t_m}=-\mathcal B^*_m\psi^*,
\end{align}
where $\psi^*$ is the associated eigenfunction.
By imposing the symmetry constraint \cite{CL-JPA-1992}
\begin{align}\label{v-sym}
	\hat v=\psi\psi^*=\tilde p\tilde q,~~~~(\tilde p,\tilde q):=(\psi,\psi^*),
\end{align}
the operator $\mathcal L$ can be represented as
\begin{align}
	\mathcal L=\partial_x+\tilde p\partial_x^{-1}\tilde q\partial_x,
\end{align}
and meanwhile, the coupled system $\psi_{t_n}=\mathcal B_n\psi$ and $\psi^*_{t_n}
=-\mathcal B_n^*\psi^*$
is restricted to \cite{Chen-2002}
\begin{subequations}
	\begin{align}
		&\tilde p_{t_n}=\mathcal B_n\tilde p=\left((\partial_x+\tilde p\partial_x^{-1}\tilde q\partial_x)^n\right)_{\geq1}\tilde p, \\
		&\tilde q_{t_n}=-\mathcal B_n^*\tilde q=-\left((-\partial_x+\partial_x\tilde q\partial_x^{-1}
\tilde p)^n\right)_{\geq1}\tilde q.
	\end{align}
\end{subequations}
This gives rise to the CLL hierarchy (for $n\geq1$)
\begin{align}\label{CLL-h}
	\begin{pmatrix}
		\tilde p \\
		\tilde q
	\end{pmatrix}_{t_{n}}
	=\widetilde{\mathcal R}^n \widetilde{\mathcal H}_0, ~~~~\widetilde{\mathcal H}_0=
	\begin{pmatrix}
		\tilde p \\
		-\tilde q
	\end{pmatrix},
\end{align}
where the recursion operator $\widetilde{\mathcal R}$ is given by  \eqref{CLL-recursion}.
This fact indicates that once we have a solution pair $(\tilde p,\tilde q)$ of the CLL hierarchy with $n\geq1$,
then $\hat v=\tilde p\tilde q$  yields a solution for the mKP hierarchy \eqref{mKP-hierarchy}.

\section{The anti-self-dual Yang-Mills reduction}\label{sec-3}
In the previous section, we briefly reviewed the hierarchy structures for the ASDYM equation
along with several classical integrable systems.
Though these hierarchies have many closed connections, they are often treated separately.
For example, the Riccati-type Miura transformation \eqref{Riccati-Miura-transformation} reveals the connection
between the AKNS hierarchy and the GI hierarchy,
but it is difficult to determine $p$ from given $(r,q)$ by using \eqref{Riccati-Miura-transformation}.
On the other hand, once we have a solution of the GI hierarchy,
the associated solutions for the GDNLS hierarchy can be calculated in principle
by using the gauge transformation \eqref{gauge-GI-GDNLS}.
However, the gauge factor in \eqref{gauge-factor} involves an integration,
which makes the calculation complicated in multi-soliton case.
In this section, we demonstrate how the AKNS and DNLS hierarchies emerge from
the gauge potential expression
\eqref{gauge-potential} and the $K$-matrix formulation \eqref{K-matrix-formulation}.
As a result, these complex transformations find a perfect realization in the same framework of the ASDYM reduction.

\subsection{Dimensional reduction condition}\label{sec-3-1}

The ASDYM equations are equipped with a group structure, which is usually characterized by
$G=\mathrm{GL}(M,\mathbb C)$, where $M\geq2$.
Here we assume $M=m_1+m_2$ and partition the $J$-matrix and $K$-matrix as $2\times2$ block matrices:
\begin{align}\label{matrix-block}
	J=
	\begin{pmatrix}
		(J_{11})_{m_1\times m_1} & (J_{12})_{m_1\times m_2} \\
		(J_{21})_{m_2\times m_1} & (J_{22})_{m_2\times m_2}
	\end{pmatrix},~~~~
	K=
	\begin{pmatrix}
		(K_{11})_{m_1\times m_1} & (K_{12})_{m_1\times m_2} \\
		(K_{21})_{m_2\times m_1} & (K_{22})_{m_2\times m_2}
	\end{pmatrix}.
\end{align}
In this sense, the expansion \eqref{gauge-potential} indicates  $\partial_{n+1}J=-(\partial_nK)J$,
which yields
\begin{subequations}\label{J-tn+1}
	\begin{align}
		\partial_{n+1}J_{11}=-(\partial_{n}K_{11})J_{11}-(\partial_{n}K_{12})J_{21}, \\
		\partial_{n+1}J_{12}=-(\partial_{n}K_{11})J_{12}-(\partial_{n}K_{12})J_{22}, \\
		\partial_{n+1}J_{21}=-(\partial_{n}K_{21})J_{11}-(\partial_{n}K_{22})J_{21}, \\
		\partial_{n+1}J_{22}=-(\partial_{n}K_{21})J_{12}-(\partial_{n}K_{22})J_{22}.
	\end{align}
\end{subequations}

The dimensional reduction is imposed by assuming the following constraint on $K$-matrix:
\begin{align}\label{dimensional-reduction}
	\partial_{0}K=[K,E],
\end{align}
where
\begin{align}\label{3.4}
E:=\frac{1}{m_1+m_2}
	\begin{pmatrix}
		m_2(I_{m_1\times m_1}) & 0 \\
		0 & -m_1(I_{m_2\times m_2})
	\end{pmatrix}.
\end{align}
Explicitly, this implies the expansion of $\partial_{0} K$ and $\partial_{1}J=-(\partial_0K)J$ gives rise to
\begin{align}\label{K-t0}
	\partial_{0}K=
	\begin{pmatrix}
		0 & -K_{12} \\
		K_{21} & 0
	\end{pmatrix},
\end{align}
and
\begin{align}\label{J-t1}
	\partial_1 J=
	\begin{pmatrix}
		K_{12}J_{21} & K_{12}J_{22} \\
		-K_{21}J_{11} & -K_{21}J_{12}
	\end{pmatrix}.
\end{align}
Under the reduction \eqref{dimensional-reduction},
the $K$-matrix formulation \eqref{K-matrix-formulation} with $m=0$ simplifies to
\begin{align}\label{3.5}
	\partial_{1}\partial_{n}K-[\partial_{n+1}K,E]-[\partial_{n }K,[K,E]]=0.
\end{align}
Expressing this equation in terms of each entry of $K$, we obtain the following system
\begin{subequations}\label{K-system}
	\begin{align}
		&\partial_1K_{11}=K_{12}K_{21},~~~~\partial_1K_{22}=-K_{21}K_{12}, \label{K-system-a}\\
		&\partial_1\partial_{n}K_{12}=
-\partial_{n+1}K_{12}-(\partial_{n}K_{11})K_{12}+K_{12}(\partial_{n}K_{22}),\label{K-system-b} \\
		&\partial_1\partial_{n}K_{21}=\partial_{n+1}K_{21}+(\partial_{n}K_{22})K_{21}-K_{21}(\partial_{n}K_{11}).
\label{K-system-c}
	\end{align}
\end{subequations}
Furthermore, by assuming $J_{11}$ to be invertible, the $t_1$-derivative of $J_{21}J_{11}^{-1}$ indicates
\begin{align}\label{P-derivetive}
	\partial_1(J_{21}J_{11}^{-1})=-K_{21}-(J_{21}J_{11}^{-1})K_{12}(J_{21}J_{11}^{-1}),
\end{align}
where relations \eqref{J-t1} are used in the calculation.
The two relations, i.e. \eqref{3.5} and \eqref{P-derivetive},
will reveal the construction of the matrix AKNS and GI hierarchies in the GL($M,\mathbb C$) ASDYM reduction framework.
Let us proceed.

\subsection{Reduction to the matrix AKNS and GI hierarchies}\label{sec-3-2}

\begin{theorem}\label{Th-1}
For the GL($M,\mathbb C$) ASDYM hierarchy  \eqref{K-matrix-formulation}
that satisfies the constraint \eqref{dimensional-reduction},
set  $x:=t_1$ as the spacial variable and define
\begin{equation}\label{QR-K}
(R,Q):=(K_{21},-K_{12}).
\end{equation}
Then, $(R,Q)$ satisfy  the recursive representation of the matrix AKNS hierarchy
	\begin{subequations}\label{mc-AKNS}
		\begin{align}
			\label{mc-AKNS-1}&R_{t_{n+1}}=R_{x,t_n}-R\partial_x^{-1}(QR)_{t_n}
-\partial_x^{-1}(RQ)_{t_n}R,~~~~ R_{t_0}=R,\\
			\label{mc-AKNS-2}&Q_{t_{n+1}}=-Q_{x,t_n}+Q\partial_x^{-1}(RQ)_{t_n}
+\partial_x^{-1}(QR)_{t_n}Q,~~~~ Q_{t_0}=-Q,
		\end{align}
	\end{subequations}
where $R$ and $Q^T$ are $m_1\times m_2$ matrix-valued functions.
\end{theorem}

\begin{proof}
	This theorem is a direct result of \eqref{K-system}.
In fact, \eqref{K-system-a} indicates
$K_{11}=\partial_x^{-1}(K_{12}K_{21})$ and $K_{22}=-\partial^{-1}_x(K_{21}K_{12})$.
Inserting them into \eqref{K-system-b} and \eqref{K-system-c}
and expressing the results in terms of $(R,Q)$ using \eqref{QR-K},
we get \eqref{mc-AKNS}.
\end{proof}

The system \eqref{mc-AKNS} appears to be a non-commutative version of \eqref{AKNS}, and can be viewed as the multi-component AKNS hierarchy as well.
For example, taking $n=1$ in \eqref{mc-AKNS} yields a matrix form of \eqref{NLS-couple}
\begin{subequations}\label{mc-NLS-couple}
	\begin{align}
		&R_{t_2}=R_{xx}-2RQR,\\
		&Q_{t_2}=-Q_{xx}+2QRQ.
	\end{align}
\end{subequations}
Defining $(\zeta,\tau):=(\mathrm ix,\mathrm it_2)$ and imposing $U:=R=\delta Q^\dagger$, where `$^\dagger$' denotes the complex conjugate transpose, it gives rise to the matrix NLS equation
\begin{align}
	\mathrm iU_\tau+U_{\zeta\zeta}+2\delta UU^\dagger U=0,~~~~\delta=\pm1.
\end{align}
In particular, when $U$ is a two-dimensional vector, it becomes the Manakov system \cite{Manakov-1974};
when $U$ is a $2\times2$ symmetric matrix, it becomes an integrable model of the spin-1 Gross-Pitaevskii equation that describes spinor Bose-Einstein condensates \cite{Ieda-2004,Kawaguchi-2012,LSS-GP-2024}.

To construct the matrix GI hierarchy, by defining
\begin{equation}\label{P-J}
P:=J_{21}J_{11}^{-1}
\end{equation}
and using \eqref{QR-K}, from the relation \eqref{P-derivetive} we have
\begin{align}\label{P-MT}
	P_x=-R+PQP,
\end{align}
which is the matrix version of the Miura transformation \eqref{Riccati-Miura-transformation}
and connects the matrix AKNS and GI hierarchies.
Then we present the following theorem.

\begin{theorem}\label{Th-2}
For the GL($M,\mathbb C$) ASDYM hierarchy, we consider the compatible gauge potential expression  \eqref{gauge-potential} and the $K$-matrix formulation \eqref{K-matrix-formulation}
that satisfy the constraint \eqref{dimensional-reduction} where $M=m_1+m_2$.
Setting $x:=t_1$ as the spacial variable and defining
\begin{equation}\label{PQ-J}
(P,Q):=(J_{21}J_{11}^{-1},-K_{12}),
\end{equation}
we have   the recursive representation of the matrix GI hierarchy
\begin{subequations}\label{mc-GI}
\begin{align}
			\label{mc-GI-1}
& P_{t_{n+1}}
=P_{x,t_n}-\partial_x^{-1}(PQ_x+PQPQ)_{t_n}P-P\partial_x^{-1}(Q_xP+QPQP)_{t_n},~~~~P_{t_0}=P, \\
			\label{mc-GI-2}
& Q_{t_{n+1}}=-Q_{x,t_n}-\partial_x^{-1}(QP_x-QPQP)_{t_n}Q-Q\partial_x^{-1}(P_xQ-PQPQ)_{t_n},
~~~~ Q_{t_0}=-Q,
\end{align}
\end{subequations}
where $P$ and $Q^T$ are $m_1\times m_2$ matrix-valued functions.
\end{theorem}

\begin{proof}
We make use of the Miura transformation \eqref{P-MT}
which we rewrite as
\begin{equation}\label{R-MT}
R=-P_x+PQP.
\end{equation}
Inserting it into \eqref{mc-AKNS-2}  yields   \eqref{mc-GI-1} directly.
To get  \eqref{mc-GI-2},  we calculate the $t_{n+1}$-derivative of $P$ by using \eqref{J-tn+1},
which gives rise to
	\begin{align*}
		P_{t_{n+1}}=-R_{t_n}-P\partial_x^{-1}(QR)_{t_n}-\partial_x^{-1}(RQ)_{t_n}P-PQ_{t_n}P.
	\end{align*}
Then we substitute \eqref{R-MT} into the above relation, and we are led to
\begin{align*}
		P_{t_{n+1}}&=P_{x,t_n}-P(QP)_{t_n}-(PQ)_{t_n}P-P\partial_x^{-1}(-QP_x+QPQP)_{t_n}
-\partial_x^{-1}(-P_xQ+PQPQ)_{t_n}P \\
		&=P_{x,t_n}-\partial_x^{-1}(PQ_x+PQPQ)_{t_n}P-P\partial_x^{-1}(Q_xP+QPQP)_{t_n},
	\end{align*}
which is exactly \eqref{mc-GI-2}.
We complete the proof.
\end{proof}

\subsection{Reduction to the vector DNLS hierarchies}\label{sec-3-3}

\subsubsection{The scalar case}\label{sec-3-3-1}

To investigate the multi-component extensions of other DNLS hierarchies,
we first realize the gauge transformation of the scalar GI hierarchy in the ASDYM framework, where $m_1=m_2=1$.
By Theorem \ref{Th-2}, the variables of the scalar GI hierarchy are defined as
\begin{equation}\label{pq-J}
(p,q):=(J_{21}J_{11}^{-1},-K_{12}).
\end{equation}
Noticing that from \eqref{J-t1} we have
$\partial_x J_{11}=K_{12} J_{21}$,
since $p$ and $q$ are scalars, we have
\begin{equation}\label{pq-Jx}
	qp=-K_{12}J_{21}J_{11}^{-1}=-(\partial_xJ_{11})J_{11}^{-1}=-(\ln J_{11})_x.
\end{equation}
This indicates that the gauge factor \eqref{gauge-factor} is determined by
\begin{align}\label{gauge-factor-J11}
	s=\exp\left(-\partial_x^{-1}(qp)\right)
=\exp(\partial_x^{-1}(\ln J_{11})_x)=J_{11},
\end{align}
which agrees with \eqref{s-dynamics-x} and \eqref{pq-Jx}.
Now that $p, q$ and the gauge factor $s$ are formulated in terms of the ASDYM variables $J$ and $K$,
so are the DNLS variables defined in \eqref{uv}, \eqref{pq-tilde} and \eqref{pq-hat}.
In the following we summarize the realizations of the DNLS hierarchies in the ASDYM framework in Table \ref{Tab-1}.

\begin{table}[h!tbp]
	\caption{Realizations of the DNLS hierarchies in the ASDYM framework} \label{Tab-1}
	\vskip 10pt
	\centering
	\renewcommand{\arraystretch}{1.5}
	\begin{tabular}{|l|l|}
		\hline
		 {The DNLS hierarchies} &  {Realizations in the ASDYM framework} \\
		\hline The GI hierarchy \eqref{GI} & $(p,q)=(J_{21}J_{11}^{-1},-K_{12})$ \\
		\hline The CLL hierarchy \eqref{CLL} & $(\tilde p,\tilde q)=(J_{21},-J_{11}^{-1}K_{12})$ \\
		\hline The KN hierarchy \eqref{KN} & $(\hat p,\hat q)=(J_{21}J_{11},-J_{11}^{-2}K_{12})$ \\
		\hline The GDNLS hierarchy \eqref{GDNLS} & $(u,v)=(J_{21}J_{11}^{\gamma-1},-J_{11}^{\gamma}K_{12})$ \\
		\hline The gauge factor \eqref{gauge-factor} & $s=J_{11}$ \\ \hline
	\end{tabular}
\end{table}

\subsubsection{The vector case}\label{sec-3-3-2}

The above discussion in the scalar case can be extended to the vector case,
if we still assume $J_{11}$ to be a scalar function
while $P$ and $Q$ defined in \eqref{PQ-J} are  a column vector and a row vector respectively.
This setting corresponds to the partition $M=m_1+m_2$ where $m_1=1, m_2=m$.
We now construct the vector GDNLS hierarchy.

\begin{theorem}\label{Th-3}
For the GL($M,\mathbb C$) ASDYM hierarchy, we consider the compatible gauge potential expression  \eqref{gauge-potential} and the $K$-matrix formulation \eqref{K-matrix-formulation}  that satisfy the constraint \eqref{dimensional-reduction} where $m_1=1, m_2=m$.
Set $x:=t_1$ as the spacial variable and define
\begin{equation}\label{UV-J}
(U,V):=(J_{21}J_{11}^{\gamma-1},-J_{11}^{\gamma}K_{12}).
\end{equation}
Note that in this case $J_{11}$ is a scalar while $U$ and $V^T$ are $m$-th order column vectors.
Then we have the recursive representation of the vector GDNLS hierarchy
	\begin{subequations}\label{VGDNLS}
		\begin{align}
			U_{t_{n+1}}&=U_{x,t_n}+\gamma \partial_x^{-1}(U_xV)_{t_n}U+\gamma U_x\partial_x^{-1}(VU)_{t_n} \notag \\
			&~~~+(\gamma-1)U\partial_x^{-1}(V_xU)_{t_n}+(\gamma-1)\partial_x^{-1}(UV_x)_{t_n}U \notag \\
			\label{VGDNLS-1}&~~~+(\gamma-1)\partial_x^{-1}(UVUV)_{t_n}U-(\gamma-1)^2U\partial_x^{-1}(VUVU)_{t_n},
~~~~ U_{t_0}=U, \\
			V_{t_{n+1}}&=-V_{x,t_n}+\gamma V\partial_x^{-1}(UV_x)_{t_n}+\gamma\partial_x^{-1}(VU)_{t_n}V_x \notag \\
			&~~~+(\gamma-1)\partial_x^{-1}(VU_x)_{t_n}V+(\gamma-1)V\partial_x^{-1}(U_xV)_{t_n} \notag \\
			\label{VGDNLS-2}&~~~-(\gamma-1)V\partial_x^{-1}(UVUV)_{t_n}+(\gamma-1)^2\partial_x^{-1}(VUVU)_{t_n}V,
~~~~ V_{t_0}=-V.
		\end{align}
	\end{subequations}
\end{theorem}

\begin{proof}
The proof is provided in Appendix \ref{app-A}.
\end{proof}

The system \eqref{VGDNLS} appears to be the vectorial generalization of \eqref{GDNLS}.
By defining $(\xi,\tau):=(\mathrm ix,\mathrm it_2)$ as real coordinates and imposing $V=\delta U^\dagger$,
it reduces to the vector GDNLS equation
\begin{align}\label{VGDNLS-equation}
	\mathrm iU_\tau+U_{\zeta\zeta}-2\mathrm i\gamma\delta U_{\zeta}U^\dagger U
-2\mathrm i(\gamma-1)\delta UU^\dagger_\zeta U+(\gamma-1)(\gamma-2)U(U^\dagger U)^2=0.
\end{align}
Assuming $U=(u_1,u_2,\dots,u_m)^T$, one rewrites \eqref{VGDNLS-equation} into its explicit form ($j=1,\dots,m$)
\begin{align}
	\mathrm iu_{j,\tau}+\mathrm iu_{j,\zeta\zeta}-2\mathrm i\gamma\delta u_{j,\zeta}\sum_{k=1}^m|u_k|^2-2\mathrm i(\gamma-1)\delta u_j
\sum_{k=1}^m \bar u_{k,\zeta} u_k+(\gamma-1)(\gamma-2)u_j\sum_{k=1}^m |u_k|^4=0.
\end{align}
Similar to the scalar GDNLS equation, the vector GDNLS equation \eqref{VGDNLS-equation} becomes the vector
GI DNLS equation when $\gamma=0$ (see equation (12) in \cite{Olver-1998})
\begin{align}
	\mathrm i U_\tau+U_{\zeta\zeta}-2\mathrm i\delta UU_\zeta^\dagger U+2(U^\dagger U)^2U=0,
\end{align}
and the vector CLL DNLS equation when $\gamma=1$ (see equation (10) in \cite{Olver-1998})
\begin{align}
	\mathrm i U_\tau+U_{\zeta\zeta}-2\mathrm i\delta U_\zeta U^\dagger U=0,
\end{align}
as well as the vector KN DNLS equation when $\gamma=2$ (see equation (2) in \cite{Olver-1998})
\begin{align}
	\mathrm i U_\tau+U_{\zeta\zeta}-4\mathrm i\delta U_\zeta U^\dagger U-2\mathrm i\delta UU_\zeta^\dagger U=0.
\end{align}
For more details about the multi-component DNLS equations, one can also refer to \cite{Fordy-1984,Hisakado-1994,Tsuchida-1999,Tsuchida-1999-2} and the references therein.

\section{Additional comments on GL(2) gauge group}\label{sec-4}

When the gauge group is restricted to be $G=\mathrm{GL}(2)$, the matrix AKNS and vector DNLS hierarchies
become scalar systems in the ASDYM reduction framework.
In this case, more fruitful results would emerge.

\subsection{Lax representation of the AKNS system}\label{sec-4-1}

The first result is the realization of the Lax representation of the AKNS system in the ASDYM framework.
This provides an alternative way to obtain the AKNS hierarchy
from the ASDYM reduction.
To achieve that, in addition to the dimensional reduction constraint \eqref{dimensional-reduction},
we assume $\partial_0\Psi=\sigma_3\Psi$.
Then the ASDYM linear system \eqref{ASDYM-linear-system} with positive $n$ is written as
\begin{align}\label{Lax-K}
	\Psi_x=\left(\lambda\sigma_3-[K,\frac{1}{2}\sigma_3]\right)\Psi, ~~~~
	\Psi_{t_n}=\left(\lambda^n\sigma_3-\sum_{i=0}^{n-1}\lambda^iK_{t_{n-i-1}}\right)\Psi.
\end{align}
According to \eqref{K-system}, the $K$-matrix can be expressed with respect to
\eqref{QR-K} in scalar case, i.e.
\begin{equation}\label{qr-K}
(r,q):=(K_{21},-K_{12}),
\end{equation}
which reads
\begin{align}\label{K-construction}
	K=\begin{pmatrix}
		-\partial_x^{-1}(qr) & -q \\
		r & \partial_x^{-1}(qr)
	\end{pmatrix}.
\end{align}
Substituting \eqref{K-construction} into \eqref{Lax-K}  gives rise to the Lax pair for the AKNS hierarchy \eqref{AKNS},
in which the $x$-derivative of $\Psi$ is
\begin{align}\label{Psi-x}
	\Psi_x=
	\begin{pmatrix}
		\lambda & -q \\
		-r & -\lambda
	\end{pmatrix}\Psi,
\end{align}
which is exactly the AKNS spectral problem.
For the time part, when $n=2$, the $t_2$-derivative of $\Psi$ expands to
\begin{align}\label{Psi-t2}
	\Psi_{t_2}=\left(\lambda^2\sigma_3-\lambda [K,\frac{1}{2}\sigma_3]-K_x \right)\Psi=
	\begin{pmatrix}
		\lambda^2+qr & -\lambda q+q_x \\
		-\lambda r-r_x & -\lambda^2-qr
	\end{pmatrix}\Psi.
\end{align}
The compatible condition of \eqref{Psi-x} and \eqref{Psi-t2} yields
the unreduced NLS system \eqref{NLS-couple}.
When taking $n=3$, we have
\begin{align}
	\Psi_{t_3}&=\left(\lambda^3\sigma_3-\lambda^2[K,\frac{1}{2}\sigma_3]-\lambda K_x-K_{t_2}\right)\Psi, \notag \\
	\label{Psi-t3}&=\begin{pmatrix}
		\lambda^3+\lambda qr-q_xr+qr_x & -\lambda^2q+\lambda q_x-q_{xx}+2q^2r \\
		-\lambda^2r-\lambda r_x-r_{xx}+2qr^2 & -\lambda^3-\lambda qr-qr_x+q_xr
	\end{pmatrix}\Psi,
\end{align}
which along with \eqref{Psi-x} gives rise to the coupled system for mKdV equation
\begin{subequations}
	\begin{align}
		r_{t_3}&=r_{xxx}-6qrr_x, \\
		q_{t_3}&=q_{xxx}-6qrq_x.
	\end{align}
\end{subequations}

\subsection{Reduction to the KP and mKP hierarchies}\label{sec-4-2}

In Sec.\ref{sec-2-4} and \ref{sec-2-5} we have reviewed the KP and mKP hierarchies
along with the associated squared eigenfunction symmetry constraints.
Since both the AKNS and CLL hierarchies can be formulated in the ASDYM reduction framework,
the KP and mKP hierarchies can be also obtained in the same  framework:
\begin{itemize}
	\item The KP hierarchy from the AKNS hierarchy: $\hat u=-qr=K_{12}K_{21}=K_{11,x}$.
	\item The mKP hierarchy from the CLL hierarchy: $\hat v=\tilde q\tilde p=qp=-K_{12}J_{21}J_{11}^{-1}=-(\ln J_{11})_x$.
\end{itemize}
Let us summarize the above constructions in the following theorem.

\begin{theorem}\label{Th-4}
For the GL($2,\mathbb C$) ASDYM hierarchy, we consider the $J$-matrix formulation  \eqref{J-matrix-formulation} and the $K$-matrix formulation \eqref{K-matrix-formulation}  that satisfy the constraint \eqref{dimensional-reduction}.
Set $(x,y):=(t_1,t_2)$ as the special variables.
Then, $u:=K_{11,x}$ satisfies the KP hierarchy \eqref{KP-hierarchy},
and  $v=-(\ln J_{11})_x$ satisfies the mKP hierarchy \eqref{mKP-hierarchy}.
\end{theorem}

We may address more words to explain how the KP hierarchy is constructed from
the AKNS hierarchy. In fact, in light of the symmetry construct \eqref{u-qr},
the positive AKNS hierarchy \eqref{AKNS-h} can recover
\[\phi_{t_n}=B_n\phi,~~~  \phi_{t_n}^*=-B_n^*\phi^*,~~~ (n=1,2,\cdots).\]
The KP hierarchy is defined from the compatibility of $\{\phi_{t_n}=B_n\phi\}$
for different $n$.
Likewise, the mKP hierarchy can be obtained from the CLL hierarchy and the constraint \eqref{v-sym}.

It is believed that the KP equation has not yet been derived from the ASDYM equations by reduction
in the known literature unless the existence of operator-valued gauge fields is allowed \cite{Schiff-1992-2,Ablowitz-2003}.
In this section, we have provided operator-free reductions for
both the KP and the mKP hierarchies.
In the formulation $K_{11}$ and $J_{11}$
are used, which belong to the entries of the $K$-matrix and $J$-matrix with the constraint \eqref{dimensional-reduction}.

\subsection{Bilinear transformation for the GDNLS equation}\label{sec-4-3}

The bilinear transformation for ASDYM equations has been given by \cite{Ohta-2001,Ohta-2024}:
\begin{align}\label{bilinear-transformation}
	J:=\frac{G}{f},~~~~J^{-1}:=\frac{S}{f},~~~~K:=\frac{H}{f},
\end{align}
which transforms \eqref{gauge-potential} into the following bilinear equations
\begin{align}
	D_{n+1}G\cdot S=-2D_{n} H\cdot (fI),
\end{align}
where $GS=f^2I$ and $I$ is the identity matrix. Here $D_n$ is the Hirota's bilinear derivative \cite{Hirota-book}
\begin{align}\label{bilinear-equation}
	D_nf(t_n)\cdot g(t_n)=(\partial_{n}-\partial_{n'})f(t_n)g(t_{n'})|_{n'=n}.
\end{align}

For the GL(2) ASDYM equations, where $G$ and $H$ are $2\times2$ matrices, we partition them as $G=\{g_{ij}\}$ and $H=\{h_{ij}\}$ ($1\leq i,j\leq 2$). Thus the entries in $J$ and $K$ can be rewritten as
\begin{align}
	J_{ij}=\frac{g_{ij}}{f},~~~~K_{ij}=\frac{h_{ij}}{f}.
\end{align}
According to the bilinear transformation \eqref{bilinear-transformation} and the ASDYM reduction framework, we propose the bilinear transformation for GDNLS system as
\begin{align}
	(u,v)=(J_{21}J_{11}^{\gamma-1},-J_{11}^{-\gamma}K_{12})
=\left(\frac{g_{11}^{\gamma-1}g_{21}}{f^\gamma},-\frac{f^{\gamma-1}h_{12}}{g_{11}^\gamma}\right),
\end{align}
which perfectly agrees with the  construction in \cite{Kakei-1995} by
Kakei, Sasa and Satsuma.
This fact provides an independent support
for the bilinear transformations in the ASDYM reductions.

The GDNLS hierarchy contains all three types of the DNLS hierarchies, their bilinear transformations
are summarized as follows:
\begin{itemize}
	\item The GI hierarchy \eqref{GI}: $(p,q)=(J_{21}J_{11}^{-1},-K_{12})=\left(g_{21}/g_{11},-h_{12}/f\right)$.
	\item The CLL hierarchy \eqref{CLL}: $(\tilde p,\tilde q)=(J_{21},-J_{11}^{-1}K_{12})=(g_{21}/f,-h_{12}/g_{11})$.
	\item The KN hierarchy \eqref{KN}: $(\hat p,\hat q)=(J_{21}J_{11},-J_{11}^{-2}K_{12})=(g_{11}g_{21}/f^2,-fh_{12}/g_{11}^2)$.
\end{itemize}
It is interesting to note that the FL equation has the same formulation as the GI system in this unified reduction framework (see \cite{Liu-2021}, Eq.(2.18)), although the FL system is derived as a potential form in the negative order KN hierarchy.
Our construction also explains why the bilinear transformations of the DNLS equations
have different tau functions \cite{Nakamura-1980,Wang-2022}.


\section{Special soliton solutions}\label{sec-5}

In Section 3, we have realized the AKNS hierarchy and the DNLS hierarchies in the ASDYM reduction framework. In this part, we will show the ASDYM solutions consequently give rise to the NLS and DNLS solutions.

\subsection{Solutions of the ASDYM}\label{sec-5-1}

While explicit solutions to the GL$(M,\mathbb C)$ ASDYM equations have been constructed in several works \cite{Huang-2021,Ohta-2024,LSS-2023}, we here present one kind of Gram-type solutions expressed by quasi-determinants.
Briefly speaking, the concept of quasi-determinant is a non-commutative generalization of the ratio between a matrix and its submatrix.
Consider a partitioned block matrix $H$ of the following form:
\begin{align}
	H =
	\begin{pmatrix}
		A & B \\
		C & d
	\end{pmatrix},
\end{align}
where $A\in\mathbb C_{N\times N}$ is a square invertible matrix,
$d\in\mathbb C_{m_1\times m_2}$, with $B$ and $C$ being matrices of compatible dimensions.
The quasi-determinant of $H$ with respect to the block $d$
(which is framed) is given by the formula
\begin{align}\label{quasi-D}
	|H|_{[d]}:=
	\begin{vmatrix}
		A & B \\
		C & \boxed{d}
	\end{vmatrix}=
	d-CA^{-1}B.
\end{align}
For details about properties of quasi-determinants, one can refer  to \cite{Gelfand-1991,Gelfand-2005}.

To present ASDYM solutions we first recall a result on solutions to the Sylvester equation.
\begin{lemma}\label{Lem-1} \cite{Sylvester-1884}
For given square matrices $A, B, C$ of the same order, the Sylvester equation
\begin{equation}
AX-XB=C
\end{equation}
has a unique matrix solution $X$ if $A$ and $B$ do not share eigenvalues.
\end{lemma}

Then we have the following.

\begin{theorem}\label{Th-5}
	Let $G=\text{GL}(M,\mathbb C)$, we assume $J^{[0]}$ and $K^{[0]}$ to be seed solutions that satisfy
	\begin{align}\label{5.4}
		A_{n+1}^{[0]}:=-(\partial_{n+1}J^{[0]})(J^{[0]})^{-1}=\partial_{n}K^{[0]}.
	\end{align}
	Suppose constant matrices $\Xi,\Lambda\in\mathbb C_{N\times N}$ do not share eigenvalues,
$\eta$ and $\theta$  satisfy the following linear differential system:
\begin{subequations}\label{linear-system}
\begin{align}
	&	\partial_{{n+1}}\eta-\Xi(\partial_{n}\eta)=\eta A^{[0]}_{n+1},\\
	&	\partial_{{n+1}}\theta-(\partial_{n}\theta)\Lambda=-A^{[0]}_{n+1} \theta.
	\end{align}
\end{subequations}
Introduce a square  matrix $\Omega$ as a solution of the Sylvester equation
	\begin{align}\label{Syl-eq}
		\Xi\Omega-\Omega\Lambda=\eta\theta,
	\end{align}
from which $\Omega$ is uniquely defined because $\Xi$ and $\Lambda$ do not share eigenvalues.
Then,	 the following constructions
	\begin{align}\label{VDT-ASDYM}
		J=
		\begin{vmatrix}
			\Omega & \Xi^{-1}\eta \\
			\theta & \boxed{I}
		\end{vmatrix}J^{[0]},~~~~K=-
		\begin{vmatrix}
			\Omega & \eta \\
			\theta & \boxed{0}
		\end{vmatrix}+K^{[0]}
	\end{align}
satisfy the compatible expression \eqref{gauge-potential} for the gauge potentials.
\end{theorem}
\begin{proof}
We leave the proof with details in Appendix \ref{app-B}.
A similar one is available in Section 3 of \cite{LSS-2025-2}.
\end{proof}

We are going to present explicit forms for the elements involved in  \eqref{VDT-ASDYM}.
By setting $J^{[0]}=I$ and $K^{[0]}=0$, i.e., $A_{n+1}^{[0]}=0$, the ASDYM solution in vacuum state is governed by the following Cauchy matrix structure (see also in Section 2.3 of \cite{LSS-2023}):
\begin{align}\label{Cauchy matrix structure}
	\Xi\Omega-\Omega\Lambda=\eta\theta,~~~~
	\partial_{n}\eta=\Xi^n(\partial_0\eta),~~~~\partial_n\theta=(\partial_0\theta)\Lambda^n.
\end{align}
For convenience, we choose $\Xi$ and $\Lambda$ to be diagonal matrices,
and thus the system \eqref{Cauchy matrix structure} is satisfied by the following construction:
\begin{subequations}\label{5.8}
	\begin{align}
		&\Xi:=\mathrm{Diag}(\xi_1,\cdots,\xi_N),~~~~\eta:=\left(\phi_{is}(\mathfrak L(\xi_i,\mathbf t))\right)_{N\times M}, \\
		&\Lambda:=\mathrm{Diag}(\lambda_1,\cdots,\lambda_N),~~~~\theta:=\left(\psi_{sj}(\mathfrak L(\lambda_j,\mathbf t))\right)_{M\times N}, \\
		&	\Omega=(\Omega_{ij})_{N\times N},~~~~\Omega_{ij}
=\frac{\sum_{s=1}^M\phi_{is}\psi_{sj}}{\xi_i-\lambda_j}, \label{Omega}
	\end{align}
\end{subequations}
where $\phi_{is}$ and $\psi_{sj}$ are arbitrary functions of the following linear phase factors respectively:
\begin{align}
	\mathfrak L(\xi_i,\mathbf t)=\sum_{n\in\mathbb Z}\xi_i^nt_{n},~~~~\mathfrak L(\lambda_j,\mathbf t)=\sum_{n\in\mathbb Z}\lambda_j^nt_{n}.
\end{align}
In this case, the explicit formulae of matrix formulations are given by
\begin{align}\label{5.10}
	J=\begin{vmatrix}
		\Omega & \Xi^{-1}\eta \\
		\theta & \boxed{I}
	\end{vmatrix}=I-\theta\Omega^{-1}\Xi^{-1}\eta,~~~~
	K=-\begin{vmatrix}
		\Omega & \eta \\
		\theta & \boxed{0}
	\end{vmatrix}=\theta\Omega^{-1}\eta.
\end{align}

\subsection{Realization of dimensional reduction}\label{sec-5-2}

To meet with the constraint \eqref{dimensional-reduction},
we impose the initial condition on the $t_0$-flows of $\eta$ and $\theta$:
\begin{align}\label{initial-condition}
	\partial_0\eta=\eta E,~~~~\partial_0\theta=-E\theta,~~~~
\end{align}
where $E$ is the block matrix defined in \eqref{3.4}.
In fact, taking derivative $\partial_0$ on the Sylvester equation  in \eqref{Cauchy matrix structure}
yields
\begin{align}
	\Xi(\partial_0\Omega)-(\partial_0\Omega)\Lambda=\eta E\theta -\eta E\theta=0,
\end{align}
which leads to $\partial_0\Omega=0$. Thus, from \eqref{5.10} we have
\begin{align}
	\partial_0K=(\partial_0\theta)\Omega^{-1}\eta+\theta\Omega^{-1}(\partial_0\theta)=-EK+KE=[K,E],
\end{align}
which recovers the constraint \eqref{dimensional-reduction}.

To get explicit expressions for $J_{ij}$ and $K_{ij}$ $(i,j=1,2)$, by reformulating $\eta$ and $\theta$ into
\begin{subequations}\label{5.12}
	\begin{align}
		&\eta=(\eta_1,\eta_2),~~~~\eta_1=(\phi_{is})_{N\times m_1},~~~~\eta_2=(\phi_{is})_{N\times m_2}, \\
		&\theta=\begin{pmatrix}
			\theta_1 \\ \theta_2
		\end{pmatrix},
		~~~~\theta_1=(\psi_{sj})_{m_1\times N},~~~~\theta_2=(\psi_{sj})_{m_2\times N},
	\end{align}
\end{subequations}
the elements in the  block matrices $J$ and $K$ in \eqref{matrix-block} are respectively given by
\begin{subequations}\label{J-expand}
	\begin{align}
		&J_{11}=I-\theta_1\Omega^{-1}\Xi^{-1}\eta_1,~~~~
		J_{12}=-\theta_1\Omega^{-1}\Xi^{-1}\eta_2, \\
		&J_{21}=-\theta_2\Omega^{-1}\Xi^{-1}\eta_1,~~~~
		J_{22}=I- \theta_2\Omega^{-1}\Xi^{-1}\eta_2,
	\end{align}
\end{subequations}
as well as
\begin{subequations}\label{K-expand}
	\begin{align}
		&K_{11}=\theta_1\Omega^{-1}\eta_1,~~~~K_{12}=\theta_1\Omega^{-1}\eta_2, \\
		&K_{21}=\theta_2\Omega^{-1}\eta_1,~~~~K_{22}=\theta_2\Omega^{-1}\eta_2.
	\end{align}
\end{subequations}
In addition, the initial condition \eqref{initial-condition}
indicates that $\phi_{is}$ and $\psi_{sj}$ are exponential functions of the following forms:
\begin{subequations}\label{phi-psi}
	\begin{align}
		&\phi_{is}=A_{is}\exp\left(\frac{1}{m_1+m_2}\sum_{n\in\mathbb Z}\xi_i^nt_n\right),~~~~1\leq s\leq m_1, \\
		&\phi_{is}=A_{is}\exp\left(-\frac{1}{m_1+m_2}\sum_{n\in\mathbb Z}\xi_i^nt_n\right),
~~~~m_1+1\leq s\leq m_1+m_2, \\
		&\psi_{sj}=B_{sj}\exp\left(-\frac{1}{m_1+m_2}\sum_{n\in\mathbb Z}\lambda_j^nt_n\right),
~~~~1\leq s\leq m_1, \\
		&\psi_{sj}=B_{sj}\exp\left(\frac{1}{m_1+m_2}\sum_{n\in\mathbb Z}\lambda_j^nt_n\right),
~~~~m_1+1\leq s\leq m_1+m_2,
	\end{align}
\end{subequations}
where $A_{is},B_{sj}\in\mathbb C$ are initial phase factors.
Inserting them into \eqref{5.12} and  \eqref{5.8} one can get explicit forms of $\eta_j, \theta_j$
and $\Omega$, and then get $J_{ij}$ and $K_{ij}$ from \eqref{J-expand} and \eqref{K-expand},
respectively.

\subsection{Solutions of the focusing NLS equation and GI DNLS equation}\label{sec-5-3}

As examples, we show in the following that how solutions of the
focusing NLS equation \eqref{NLS}$|_{\delta=1}$
and the focusing GI DNLS equation \eqref{GI-equation}$|_{\delta=1}$
can be obtained from the above solutions presented in Sec.\ref{sec-5-2}.

To get solutions for these two equations,
we set $G=\text{GL}(2)$ and define $(\zeta,\tau):=(\mathrm it_1,\mathrm it_2)$ to be real coordinates,
 which indicates $m_1=m_2=1$ and \eqref{phi-psi} becomes
\begin{subequations}\label{phi-psi-GL2}
	\begin{align}
		&\phi_{i1}=A_{i1}\exp\left(-\mathrm i\frac{\xi_i}{2}\zeta-\mathrm i\frac{\xi_i^2}{2}\tau\right),~~~~
		\phi_{i2}=A_{i2}\exp\left(\mathrm i\frac{\xi_i}{2}\zeta+\mathrm i\frac{\xi_i^2}{2}\tau\right), \label{5.19a} \\
		&\psi_{1j}=B_{1j}\exp\left(\mathrm i\frac{\lambda_j}{2}\zeta+\mathrm i\frac{\lambda_j^2}{2}\tau\right),~~~~
		\psi_{2j}=B_{2j}\exp\left(-\mathrm i\frac{\lambda_j}{2}\zeta-\mathrm i\frac{\lambda_j^2}{2}\tau\right).
	\end{align}
\end{subequations}
$\eta_j, \theta_j$ and $\Omega$ are correspondingly defined from
\eqref{5.12} and \eqref{Omega}.

\subsubsection{Solutions of the focusing NLS equation}\label{sec-5-3-1}

In light of Theorem \ref{Th-1} and \eqref{K-expand}, we have
\begin{equation}
(r,q)=(\theta_2\Omega^{-1}\eta_1,-\theta_1\Omega^{-1}\eta_2),
\end{equation}
which provide solutions to the  following unreduced NLS system:
\begin{align}
	\mathrm ir_{\tau}+r_{\zeta\zeta}+2qr^2=0,~~~~\mathrm iq_{\tau}-q_{\zeta\zeta}-2q^2r=0.
\end{align}
To achieve solutions for the focusing NLS equation, we introduce the conjugate constraints as
\begin{align}
	\Lambda=\Xi^\dagger,~~~~
A_{i1}=\bar B_{1i}, ~~~~  A_{i2}=\bar B_{2i}, ~~~~(i=1,2,\cdots N).
\end{align}
It follows that
\begin{align}
\theta_1=\eta_1^\dagger,~~~~\theta_2=\eta_2^\dagger,
\end{align}
where
\begin{align}\label{eta-expression}
	\eta_j=(\phi_{1j}, \phi_{2 j}, \cdots, \phi_{N j})^T, ~~~~ j=1,2,
\end{align}
with $\{\phi_{kj}\}$ given in \eqref{5.19a}.
Substituting them into the Sylvester equation \eqref{Syl-eq} yields
\begin{align}
	\Xi\Omega-\Omega\Xi^\dagger=\eta_1\eta_1^\dagger+\eta_2\eta_2^\dagger.
\end{align}
Comparing it with its complex conjugate transpose and making use of Lemma \ref{Lem-1},
we have $\Omega=-\Omega^\dagger$. Finally we reach to
\begin{align}
	q=- \eta_1^\dagger\Omega^{-1}\eta_2 =(\eta_2^\dagger\Omega^{-1}\eta_1)^\dagger=\bar r.
\end{align}
Thus the above $q$  solves the focusing NLS equation
\begin{align}
	\mathrm iq_\tau+q_{\zeta\zeta}+2|q|^2q=0.
\end{align}

\subsubsection{Solutions of the focusing GI DNLS equation}\label{sec-5-3-2}

Likewise, by Theorem 2 and \eqref{K-expand}, we have
\begin{align}
	(p,q)=\left(\frac{J_{21}}{J_{11}},-K_{12}\right)
=\left(-\frac{\theta_2\Omega^{-1}\Xi^{-1}\eta_1}{1-\theta_1\Omega^{-1}\Xi^{-1}\eta_1},-\theta_1\Omega^{-1}
\eta_2\right)
\end{align}
that satisfy the unreduced GI system
\begin{align}
	\mathrm ip_\tau+p_{\zeta\zeta}+2\mathrm ip^2q_\zeta+2p^3q^2=0, ~~~~
	\mathrm iq_\tau-q_{\zeta\zeta}+2\mathrm i q^2p_\zeta-2p^2q^3=0.
\end{align}
Then we impose the conjugate constraints
\begin{align}\label{5.30}
	\Lambda=\Xi^\dagger,~~~~A_{i1}=\bar B_{1i}, ~~~~ A_{i2}=-\bar \xi_i\bar B_{2i}, ~~~~(i=1,2,\cdots,N).
\end{align}
First, it allows us to have
\begin{align}\label{5.31}
	\theta_1=\eta_1^\dagger,~~~~\theta_2=-\eta_2^\dagger\Xi^\dagger,
\end{align}
where $\eta_j$ is given in \eqref{eta-expression}.
Second, by substituting \eqref{5.30} and \eqref{5.31}  into the Sylvester equation \eqref{Syl-eq}, we are led to
\begin{align}\label{Syl-GI}
	\Xi\Omega-\Omega\Xi^\dagger=\eta_1\eta_1^\dagger-\eta_2\eta_2^\dagger\Xi^\dagger.
\end{align}
Then, subtracting  $\eqref{Syl-GI}^\dagger$ from $\eqref{Syl-GI}$ and making use of Lemma \ref{Lem-1}
show that
\begin{align}\label{Omega-Omega-dagger}
	\Omega=-\Omega^\dagger+\eta_2\eta_2^\dagger,
\end{align}
which provides the relation between $\Omega$ and $\Omega^\dagger$ for the GI DNLS equation.
Note that after substituting \eqref{Omega-Omega-dagger} into the term $\Omega\Xi^\dagger$
in \eqref{Syl-eq},  we arrive at
\begin{align}
	\Xi\Omega+\Omega^\dagger\Xi^\dagger=\eta_1\eta_1^\dagger,
\end{align}
which indicates a conjugate constraint between $\Omega$ and $\Xi$.

Thus we can rewrite the $J_{11},J_{21}$ and $K_{12}$
in terms of only $\eta_j, \Omega$ and $\Xi$:
\begin{align}\label{5.35}
	J_{11}=1-\eta_1^\dagger\Omega^{-1}\Xi^{-1}\eta_1, ~~~~
	J_{21}=\eta_2^\dagger\Xi^\dagger\Omega^{-1}\Xi^{-1}\eta_1, ~~~~
	K_{12}=\eta_1^\dagger\Omega^{-1}\eta_2.
\end{align}
Then, through a direct calculation, we have
\begin{align}
	K_{12}^\dagger J_{11}&=\eta_2^\dagger\Omega^{-\dagger}\eta_1
-\eta_2^\dagger\Omega^{-\dagger}\eta_1\eta_1^\dagger\Omega^{-1}\Xi^{-1}\eta_1 \notag \\
	&=\eta_2^\dagger\Omega^{-\dagger}\eta_1-\eta_2^\dagger(\Omega^{-\dagger}\Xi
+\Xi^\dagger\Omega^{-1})\Xi^{-1}\eta_1=-J_{21},
\end{align}
where $A^{-\dagger}$ denotes the inverse of $A^\dagger$, and $s=J_{11}$ yields
\begin{align}
	J_{11}J_{11}^\dagger&=(1-\eta_1^\dagger\Omega^{-1}\Xi^{-1}\eta_1)
(1-\eta_1^\dagger\Xi^{-\dagger}\Omega^{-\dagger}\eta_1) \notag \\
	&=1-\eta_1^\dagger(\Omega^{-1}\Xi^{-1}+\Xi^{-\dagger}\Omega^{-\dagger})\eta_1
	+\eta_1^\dagger\Omega^{-1}\Xi^{-1}\eta_1\eta_1^\dagger\Xi^{-\dagger}\Omega^{-\dagger}\eta \notag \\
	&=1-\eta_1^\dagger(\Omega^{-1}\Xi^{-1}+\Xi^{-\dagger}\Omega^{-\dagger}-\Omega^{-1}\Xi^{-1}
-\Xi^{-\dagger}\Omega^{-\dagger})\eta_1=1.
\end{align}
Thus we arrive at $p=J_{21}J_{11}^{-1}=-K_{12}^\dagger=\bar q$, which gives rise to the GI equation
\begin{align}
	\mathrm ip_\tau+p_{\zeta\zeta}+2\mathrm ip^2\bar p_\zeta+2|p|^4p=0.
\end{align}
And $u=s^\gamma p=(\bar s)^{-\gamma}\bar q$ solves the GDNLS equation \eqref{GDNLS-equation} with $\delta=1$,
where $s=J_{11}$ is given in \eqref{5.35}.

\section{Concluding remarks}\label{sec-6}

In this paper, we have derived the matrix AKNS hierarchy, the matrix GI hierarchy, the vector GDNLS hierarchy (which includes the GI, KN, and CLL hierarchies as special cases), and the scalar KP and mKP hierarchies as reductions of the ASDYM equations.
These reductions were achieved by utilizing the compatible expression \eqref{gauge-potential} for the gauge potentials $\{A_j\}$, along with the resulting equations \eqref{J-matrix-formulation} and \eqref{K-matrix-formulation}---specifically, the $J$-matrix and $K$-matrix formulations of the ASDYM equations.
All the dimensional reductions are achieved under the same constraint \eqref{dimensional-reduction}
imposed on $K$,
and the above mentioned lower dimensional hierarchies are formulated using the entries of $K$ and $J$.
This approach provides a unified framework for understanding these integrable hierarchies, their Miura-type links, and the bilinearization of DNLS equations.
The variable formulations
in GL(2) gauge group are the following:
\begin{itemize}
	\item The AKNS hierarchy \eqref{AKNS}: $(r,q)=(K_{21},-K_{12})$.
	\item The GI hierarchy \eqref{GI}: $(p,q)=(J_{21}J_{11}^{-1},-K_{12})$.
	\item  The gauge factor between DNLS hierarchies \eqref{gauge-factor}: $s=J_{11}$.
	\item The GDNLS hierarchy \eqref{GDNLS}: $(u,v)=(J_{21}J_{11}^{\gamma-1},-J_{11}^{-\gamma}K_{12})$.
	\item The KP  hierarchy \eqref{KP-hierarchy}: $u=K_{11,x}$.
    \item The mKP hierarchy \eqref{mKP-hierarchy}: $v=-(\ln J_{11})_x$.
\end{itemize}
Among them, the AKNS and GI hierarchies admit matrix version in $\text{GL}(m_1+m_2)$ gauge group, while the GDNLS hierarchy admits a vector version in $\text{GL}(1+m)$ gauge group.
Note that the vector GDNLS hierarchy reduces to the GI, CLL and KN
when $\gamma=0,1,2$, respectively.
Furthermore, it is noteworthy that we have derived the KP and mKP hierarchies via operator-free ASDYM reductions, in contrast to previous approaches \cite{Schiff-1992-2,Ablowitz-2003}.

Within this unified framework,
not only equations, but also  soliton solutions of the ASDYM equations can be systematically mapped to those of
the reduced hierarchies.
The fact about solutions has been demonstrated in Section \ref{sec-5}. We presented Gram-type solutions in terms of quasi-determinants in Theorem \ref{Th-5} and its proof is provided in Appendix \ref{app-B}.
Soliton solutions as special case of Theorem \ref{Th-5} were presented in Sec.\ref{sec-5-1} as well.
Explicit formulae of $K$ and $J$ that  meet   the constraint \eqref{dimensional-reduction}
have been given in Sec.\ref{sec-5-2},
which provide soliton solutions for all the lower dimensional hierarchies obtained through the reductions.
In addition, we also provided two examples to show further complex reductions
of solutions for the focusing NLS and GI DNLS equations.

We end up this paper with mentioning some potential topics inspired from our current work.
Firstly, in contrast to traditional ASDYM reduction techniques \cite{Ablowitz-2003,Mason-Book},
our method
allows a direct correspondence between the variables of the ASDYM equations and those of the resulting integrable hierarchies in reductions.
This means many known solutions of the ASDYM equations might be
used to generated solutions for the reduced integrable hierarchies obtained in this paper,
and vice versa.
Secondly, it is well known that the ASDYM equation has a special type solutions,
namely, instantons. The solutions we have presented in Theorem \ref{Th-5}
do not include such type of solutions.
Whether the reduction links are helpful in constructing instantons
will be considered in the future.
Thirdly, the reductions to the DNLS hierarchies and soliton solution formulations
presented in Sec.\ref{sec-5-2} implies the Cauchy matrix structures of the DNLS type equations,
which have not yet been revealed in the literature.
We will investigate these structures elsewhere.
In addition, the reductions also imply possibility to get Bogoyavlenskii's type
breaking soliton solutions. This will be considered separately.
Finally, while satisfactory discretizations for the ASDYM equation are currently lacking,
but most of the above listed lower dimensional equations have been discretized.
The results of this paper should be helpful in finding an integrable discrete analogue of the ASDYM equation.
Considering that the ASDYM equation has very rich structures in differential geometry,
a successful  discretization for the ASDYM equation will bring more insight
for differential geometry from a discrete perspective.

\vskip 20pt
\subsection*{Acknowledgments}

The work of SSL is supported by  the JSPS Overseas Research Fellowships (P25326).
The work of KIM is supported by the JSPS grant (JP22K03441, 23K22407) and the JSPS Bilateral Program (JPJSBP120247414).
The work of DJZ is supported by the NSFC grant (12271334, 1241540016).

\vskip 20pt

\begin{appendices}

\section{Proof of Theorem \ref{Th-3}}\label{app-A}

In the vector case, the gauge factor $s$ is defined as (cf.\eqref{gauge-factor})
\begin{align}\label{gauge-factor-QP}
	s=\exp\left(-\partial_x^{-1}(QP)\right),
\end{align}
where $P$ and $Q^T$ are $m$-th order column vector, defined in \eqref{PQ-J}.
Similar to the scalar case presented in in Sec.\ref{sec-3-3-1}, in the vector case,  we also have
$QP=-(\ln J_{11})_x$ and $s=J_{11}$.
Thus we can compute the dynamics of $s$ with respect to $P$ and $Q$:
\begin{subequations}
\begin{align}
&\frac{s_x}{s}=-QP, \label{s-PQ-a}
 ~~~~\frac{s_{t_n}}{s}=-\partial_x^{-1}(QP)_{t_n}, \\
&\frac{s_{t_{n+1}}}{s}=-\partial_x^{-1}(QP)_{t_{n+1}}
=-\partial_x^{-1}(QP_x)_{t_n}+\partial_x^{-1}(QPQP)_{t_n}+Q_{t_n}P,
\end{align}
\end{subequations}
where to get the last equation we have made use of \eqref{mc-GI}.
With these in hand,  we then calculate the derivatives of $U$ and $V$ defined in  \eqref{UV-J}.
Since $J_{11}$ is a scalar, we have
$U=J_{21}J_{11}^{\gamma-1}=s^\gamma P$.
Thus we find
\[U_x=s^{\gamma}P_x+\gamma s^{\gamma-1} s_x P=s^{\gamma}(P_x-\gamma PQP),\]
where we have used \eqref{s-PQ-a}.
Similarly, we have
\begin{align*}
		&U_x=s^{\gamma}(P_x-\gamma PQP), \\
		&U_{x,t_n}=s^\gamma(P_{x,t_n}-\gamma (PQP)_{t_n}-\gamma(P_x-\gamma PQP)\partial_x^{-1}(QR)_{t_n}), \\
		&U_{t_{n+1}}=s^\gamma(P_{t_{n+1}}-\gamma P\partial_x^{-1}(QP_x)_{t_n}
+\gamma P\partial_x^{-1}(QPQP)_{t_n}+\gamma PQ_{t_n}P),
	\end{align*}
as well as the derivatives of $V=-J_{11}^{-\gamma}K_{12}=s^{-\gamma}Q$:
\begin{align*}
		&V_{x,t_n}=s^{-\gamma}(Q_{x,t_n}+\gamma(QPQ)_{t_n}+\gamma(Q_x+\gamma QPQ)\partial_x^{-1}(QR)_{t_n}), \\
		&V_{t_{n+1}}=s^{-\gamma}(Q_{t_{n+1}}+\gamma \partial_x^{-1}(QP_x)_{t_n}Q-\partial_x^{-1}(QPQP)_{t_n}Q-Q_{t_n}PQ).
	\end{align*}
As for the nonlinear terms in \eqref{VGDNLS-1}, they are expressed in terms of $P$ and $Q$ as
	\begin{align*}
		&\partial_x^{-1}(U_xV)_{t_n}U=s^\gamma \partial_x^{-1}(P_xQ-\gamma PQPQ)_{t_n}P, \\
		&U_x\partial_x^{-1}(VU)_{t_n}=s^{\gamma}(P_x-\gamma PQP)\partial_x^{-1}(QP)_{t_n}, \\
		&U\partial_x^{-1}(V_xU)_{t_n}=s^\gamma P\partial_x^{-1}(Q_xP+\gamma QPQP)_{t_n}, \\
		&\partial_x^{-1}(UVUV)_{t_n}U=s^\gamma \partial_x^{-1}(PQPQ)_{t_n}P, \\
		&U\partial_x^{-1}(VUVU)_{t_n}=s^\gamma P\partial_x^{-1}(QPQP)_{t_n}. \\
	\end{align*}
Then   substituting the above expressions into the both sides of \eqref{VGDNLS-1}
and making use of \eqref{mc-GI-1}, we find \eqref{VGDNLS-1} holds.
The proof for \eqref{VGDNLS-2} can be done along the same line.
We skip presenting details.

\section{Proof of Theorem \ref{Th-5}}\label{app-B}
	
The proof consists of the following lemmas.

\begin{lemma}\label{lem-A1}
\cite{Gilson-2007,LSS-2025-2}
Consider the quasi-determinant \eqref{quasi-D}
and suppose the derivative of $A$ has the decomposition
		\begin{align}
			\partial_n A=\sum_{l=1}^nE_lF_l,
		\end{align}
where $\partial_n$ is the differential operator with respect to $t_n$,
$E_l$ and $F_l$ are matrices of compatible dimensions.
Then, the derivative of the quasi-determinant is given by
\begin{align}\label{B.2}
			\partial_n
			\begin{vmatrix}
				A & B \\
				C & \boxed{d}
			\end{vmatrix}
			=\partial_n d+
			\begin{vmatrix}
				A & B \\
				\partial_n C & \boxed{0}
			\end{vmatrix}+
			\begin{vmatrix}
				A & \partial_n B \\
				C & \boxed{0}
			\end{vmatrix}+
			\sum_{l=1}^n
			\begin{vmatrix}
				A & E_l \\
				C & \boxed{0}
			\end{vmatrix}
			\begin{vmatrix}
				A & B \\
				F_l & \boxed{0}
			\end{vmatrix}.
		\end{align}
\end{lemma}
	
\begin{lemma}\label{lem-A2}
The matrix $\Omega$ defined in \eqref{Syl-eq} satisfies the following differential recurrence relations:
\begin{align}\label{Omega-xn+1}
\partial_{n+1}\Omega
=(\partial_{n}\eta)\theta+(\partial_{n}\Omega)\Lambda
=\Xi(\partial_{n}\Omega)-\eta(\partial_{n}\theta),~~~~n\in\mathbb Z.
\end{align}
\end{lemma}

\begin{proof}
We begin by taking   $t_{n+1}$-derivative of the Sylvester equation \eqref{Syl-eq}, which gives rise to
\begin{align*}
			\Xi(\partial_{n+1}\Omega)-(\partial_{n+1}\Omega)\Lambda
=(\partial_{n+1}\eta)\theta+\eta(\partial_{n+1}\theta).
		\end{align*}
		Substituting the dispersion relations $	\partial_{n+1}\eta=\Xi(\partial_{n}\eta)$ and $\partial_{n+1}\theta=(\partial_{n}\theta)\Lambda$ into the right-hand side yields
		\begin{align*}
			\Xi(\partial_{n+1}\Omega)-(\partial_{n+1}\Omega)\Lambda=&\Xi(\partial_{n}\eta)\theta
+\eta(\partial_{n}\theta)\Lambda +\Xi(\partial_{n}\Omega)\Lambda-\Xi(\partial_{n}\Omega)\Lambda\\
			=&\Xi((\partial_{n}\eta)\theta+(\partial_{n}\Omega)\Lambda)-(\Xi(\partial_{n}\Omega)
-\eta(\partial_{n}\theta))\Lambda.
		\end{align*}
Noticing that $(\partial_{n}\eta)\theta+(\partial_{n}\Omega)\Lambda
=\Xi(\partial_{n}\Omega)-\eta(\partial_{n}\theta)$ is nothing
but the result of the $t_{n}$-derivative of \eqref{Syl-eq},
we are led  to the recursive relation \eqref{Omega-xn+1}.
\end{proof}

\begin{lemma}\label{lem-UV}
Let \eqref{VDT-ASDYM} be written as
\begin{equation}\label{JK-UV}
J=WJ^{[0]},~~~  K=Z+K^{[0]},
\end{equation}
where $W$ and $Z$ are defined by
		\begin{align}\label{B.5}
			W=
			\begin{vmatrix}
				\Omega & \Xi^{-1}\eta \\
				\theta & \boxed{I}
			\end{vmatrix},~~~~
			Z=-
			\begin{vmatrix}
				\Omega & \eta \\
				\theta & \boxed{0}
			\end{vmatrix}.
		\end{align}
		Then, these matrices satisfy the following relation:
		\begin{align}\label{V-n+1}
			\partial_{n+1}W+(\partial_nZ)W=[W,A_{n+1}^{[0]}].
		\end{align}
	\end{lemma}

\begin{proof}
By using the derivative formula \eqref{B.2} of quasi-determinant
and the recursive relation \eqref{Omega-xn+1}, we have
		\begin{align*}
			\partial_{n+1}W=
			\begin{vmatrix}
				\Omega & \Xi^{-1}\eta \\
				\partial_{n+1}\theta & \boxed{0}
			\end{vmatrix}+
			\begin{vmatrix}
				\Omega & \Xi^{-1}\partial_{n+1}\eta \\
				\theta & \boxed{0}
			\end{vmatrix}+
			\begin{vmatrix}
				\Omega & \partial_n\eta \\
				\theta & \boxed{0}
			\end{vmatrix}
			\begin{vmatrix}
				\Omega & \Xi^{-1}\eta \\
				\theta & \boxed{0}
			\end{vmatrix}+
			\begin{vmatrix}
				\Omega & \partial_n\Omega \\
				\theta & \boxed{0}
			\end{vmatrix}
			\begin{vmatrix}
				\Omega & \Xi^{-1}\eta \\
				\Lambda & \boxed{0}
			\end{vmatrix}.
		\end{align*}
For the first two terms on the right-hand side,
substituting \eqref{linear-system} into them yields
		\begin{align*}
			\begin{vmatrix}
				\Omega & \Xi^{-1}\eta \\
				\partial_{n+1}\theta & \boxed{0}
			\end{vmatrix}=
			\begin{vmatrix}
				\Omega & \Xi^{-1}\eta \\
				(\partial_n\theta)\Lambda & \boxed{0}
			\end{vmatrix}-
			A_{n+1}^{[0]}
			\begin{vmatrix}
				\Omega & \Xi^{-1}\eta \\
				\theta & \boxed{0}
			\end{vmatrix},
		\end{align*}
		and
		\begin{align*}
			\begin{vmatrix}
				\Omega & \Xi^{-1}\partial_{n+1}\eta \\
				\theta & \boxed{0}
			\end{vmatrix}=
			\begin{vmatrix}
				\Omega & \partial_n\eta \\
				\theta & \boxed{0}
			\end{vmatrix}+
			A_{n+1}^{[0]}
			\begin{vmatrix}
				\Omega & \Xi^{-1}\eta \\
				\theta & \boxed{0}
			\end{vmatrix}.
		\end{align*}
		Thus, $\partial_{n+1}W$ can be rewritten as
		\begin{align}\label{B.6}
			\partial_{n+1}W=
			\begin{vmatrix}
				\Omega & \Xi^{-1}\eta \\
				(\partial_n\theta)\Lambda & \boxed{0}
			\end{vmatrix}+
			\begin{vmatrix}
				\Omega & \partial_n\Omega \\
				\theta & \boxed{0}
			\end{vmatrix}
			\begin{vmatrix}
				\Omega & \Xi^{-1}\eta \\
				\Lambda & \boxed{0}
			\end{vmatrix}+
			\begin{vmatrix}
				\Omega & \partial_n\eta \\
				\theta & \boxed{0}
			\end{vmatrix}W+
			[W,A_{n+1}^{[0]}].
		\end{align}
		On the other hand, the direct expansion of $-(\partial_nU)V$ indicates
		\begin{align*}
			-(\partial_{n}Z)W=
			\begin{vmatrix}
				\Omega & \partial_n\eta \\
				\theta & \boxed{0}
			\end{vmatrix}W+
			\begin{vmatrix}
				\Omega & \eta \\
				\partial_n\theta & \boxed{0}
			\end{vmatrix}
			\begin{vmatrix}
				\Omega & \Xi^{-1}\eta \\
				\theta & \boxed{I}
			\end{vmatrix}+
			\begin{vmatrix}
				\Omega & \partial_n\Omega \\
				\theta & \boxed{0}
			\end{vmatrix}
			\begin{vmatrix}
				\Omega & \eta \\
				I & \boxed{0}
			\end{vmatrix}
			\begin{vmatrix}
				\Omega & \Xi^{-1}\eta \\
				\theta & \boxed{I}
			\end{vmatrix}.
		\end{align*}
Then, noting that the following relation holds for an arbitrary matrix $C$ with a suitable size,
		\begin{align*}
			\begin{vmatrix}
				\Omega & \eta \\
				C & \boxed{0}
			\end{vmatrix}
			\begin{vmatrix}
				\Omega & \Xi^{-1}\eta \\
				\theta & \boxed{I}
			\end{vmatrix}
			&=
			-C\Omega^{-1}\eta+C\Omega^{-1}\eta\theta\Omega^{-1}\Xi^{-1}\eta \\
			&=-C\Omega^{-1}\eta+C(\Omega^{-1}\Xi-\Lambda\Omega^{-1})\Xi^{-1}\eta
			=\begin{vmatrix}
				\Omega & \Xi^{-1}\eta \\
				C\Lambda & \boxed{0}
			\end{vmatrix},
		\end{align*}
thus we have
		\begin{align}\label{B.7}
			-(\partial_{n}Z)W=
			\begin{vmatrix}
				\Omega & \partial_n\eta \\
				\theta & \boxed{0}
			\end{vmatrix}W+
			\begin{vmatrix}
				\Omega & \Xi^{-1}\eta \\
				(\partial_n\theta)\Lambda & \boxed{0}
			\end{vmatrix}+
			\begin{vmatrix}
				\Omega & \partial_n\Omega \\
				\theta & \boxed{0}
			\end{vmatrix}
			\begin{vmatrix}
				\Omega & \Xi^{-1}\eta \\
				\Lambda & \boxed{0}
			\end{vmatrix}.
		\end{align}
Combining \eqref{B.6} and \eqref{B.7}, we are led to \eqref{B.5}.

\end{proof}
	
Finally, based on \eqref{JK-UV} and \eqref{5.4},
through a direct calculation we have
\begin{align*}
\partial_{n+1}{J}+(\partial_{n}K)J &
=(\partial_{n+1}W) J^{[0]} - W A^{[0]}_{n+1} J^{[0]}+(\partial_{n}Z) W J^{[0]}+ A^{[0]}_{n+1} W J^{[0]} \\
		&=(\partial_{n+1}W+(\partial_nZ)W-[W,A_{n+1}^{[0]}])J^{[0]},
\end{align*}
which vanishes due to \eqref{V-n+1}.
Thus, $J$ and $K$ satisfy \eqref{gauge-potential} and
we complete  the proof.

\end{appendices}

\vskip 20pt

\small{

\end{document}